\documentclass{llncs}

\usepackage{xcolor,pgf,tikz,pgflibraryarrows,pgffor,pgflibrarysnakes,pgflibraryshapes}
\usepackage{latexsym,amssymb,amsmath}
\usepackage{subfigure}
\usepackage{url,xspace}
\usepackage{stmaryrd}
\usepackage{enumerate}
\usepackage{yhmath,paralist}
\usepackage[vlined,ruled,linesnumbered]{algorithm2e}

\usepackage{multirow}
\usepackage{todonotes}

\usepackage{wrapfig}
\usepackage[pdftex,colorlinks=true,bookmarks=true,bookmarksopen=true,bookmarksopenlevel=2,bookmarksdepth=2]{hyperref}

\def\Reachname{\ensuremath{{\sf Reach}}}
\def\Reach#1{\Reachname\left(#1\right)}
\def\Succname{\ensuremath{{\sf Succ}}}
\def\Succ#1{\Succname\left(#1\right)}
\def\Succlab#1#2{\Succname_{#2}\left(#1\right)}
\def\bad{\ensuremath{{\sf Bad}}}
\def\v{\ensuremath{{\sf v}}}
\def\Aa{\ensuremath{\mathcal{A}}}
\def\win{\ensuremath{{\sf Win}}}
\def\Attr{\ensuremath{{\sf Attr}}}
\def\maxwin{\ensuremath{{\sf MaxWin}}}
\def\dc#1{\ensuremath{\downarrow\!\left(#1\right)}}
\def\dco#1#2{\ensuremath{\downarrow^{#2}\!\left(#1\right)}}
\def\uc#1{\ensuremath{\uparrow\!\left(#1\right)}}
\def\uco#1#2{\ensuremath{\uparrow^{#2}\!\left(#1\right)}}
\def\mss{{\sc MinSizeStrat}\xspace}
\def\XX{\mathcal{X}}
\def\CC{\mathcal{C}}
\def\LL{\mathcal{L}}
\def\SS{\mathcal{S}}
\def\cdom#1{\ensuremath{{\sf Supp}\!\left(#1\right)}}
\def\WR#1{\ensuremath{\mathcal{W}\mathcal{R}_{#1}}}
\def\maxac#1{\ensuremath{\left\lceil #1\right\rceil}}
\def\minac#1{\ensuremath{\left\lfloor #1\right\rfloor}}
\def\hsg{\ensuremath{\hat{\sigma}_G}}

\def\badseven{%
\raisebox{-2pt}{\begin{tikzpicture}
  \node[state, fill=gray, scale=.4]  {\huge $7$};
\end{tikzpicture}} %
} \def\badeight{%
  \raisebox{-2pt}{\begin{tikzpicture} \node[state, rectangle,
      fill=gray, scale=.4] {\huge $8$};
\end{tikzpicture}} %
}%
\def\Bstate#1{%
\raisebox{-2pt}{\begin{tikzpicture}
  \node[state, rectangle,  scale=.4]   {\huge $#1$};
\end{tikzpicture}} %
}%
\def\Astate#1{%
\raisebox{-2pt}{\begin{tikzpicture}
  \node[state,  scale=.4]   {\huge $#1$};
\end{tikzpicture}} %
}%
\def\Bgraystate#1{%
\raisebox{-2pt}{\begin{tikzpicture}
  \node[state, rectangle,  scale=.4,fill=gray!20 ]   {\huge $#1$};
\end{tikzpicture}} %
}%
\def\Agraystate#1{%
\raisebox{-2pt}{\begin{tikzpicture}
  \node[state,  scale=.4, fill=gray!20]   {\huge $#1$};
\end{tikzpicture}} %
}%

\usetikzlibrary{%arrows,
automata,calc,3d,decorations.pathreplacing}
\usetikzlibrary{shapes,snakes}
\usetikzlibrary{decorations.shapes,positioning}

\def\wbs{\ensuremath{\unrhd}}
\def\nwbs{\ensuremath{\not\!\!\unrhd}}

\title{Synthesising Succinct Strategies\\ in Safety Games\thanks{This
    research has been supported by the Belgian F.R.S./FNRS FORESt
    grant, number 14621993.\hfill\break The research leading to these
    results has received funding from the European Union Seventh
    Framework Programme (FP7/2007-2013) under Grant Agreement n°601148
    (CASSTING)}}

\author{Gilles Geeraerts \and Jo\"el Goossens \and Am\'elie Stainer}
\institute{Universit\'e libre de Bruxelles, D\'epartement d'Informatique, Brussels, Belgium}

\begin{document}
\maketitle
\thispagestyle{plain}
 \begin{abstract}
   Finite turn-based safety games have been used for very different
   problems such as the synthesis of linear temporal logic
   (LTL)~\cite{FJR-fmsd11}, the synthesis of schedulers for computer
   systems running on multiprocessor platforms~\cite{BM-esa10}, and
   also for the determinisation of timed
   automata~\cite{BSJK-fossacs11}. In these contexts, games are
   implicitly defined, and their size is at least exponential in the
   size of the input. Nevertheless, there are natural relations
   between states of arenas of such games. We first formalise the
   properties that we expect on the relation between states, thanks to
   the notion of alternating simulation. Then, we show how such
   simulations can be exploited to (1) improve the running time of the
   OTFUR algorithm\cite{CDFLL-concur05} to compute winning strategies
   and (2) obtain a succinct representation of a winning strategy. We also
   show that our general theory applies to the three applications
   mentioned above.
 \end{abstract}

\section{Introduction}
Finite, turn-based, safety games are arguably one of the most simple,
yet relevant, classes of games. They are played by two players (A and
B) on a finite graph (called the arena), whose set of vertices is
partitioned into Player A and Player B vertices, (that we call $A$ and
$B$-states respectively). A play is an infinite path in this graph,
and is obtained by letting the players move a token on the
vertices. Initially, the token is on a designated initial vertex. At
each round of the game, the player who owns the vertex marked by the
token decides on which successor node to move it next. A play is
winning for $A$ if the token never touches some designated bad nodes,
otherwise, it is winning for $B$.

Such games are a natural model to describe the interaction of a
potential controller with a given environment, where the aim of the
controller is to avoid the bad states that model system failures. In
this framework, computing a winning strategy for the player amounts to
\emph{synthesising} a control policy that guarantees no bad state will
be reached, no matter how the environment behaves. Safety games have
also been used as a tool to solve other problems such as LTL
realisability \cite{FJR-fmsd11}, real-time scheduler synthesis
\cite{BM-esa10} or timed automata determinisation
\cite{BSJK-fossacs11}. We will come back to those applications.

To illustrate our ideas, we consider, throughout the paper, a
variation of the well-known Nim game \cite{Bouton190235}, where
players have to \emph{fill} an urn instead of removing balls from
it. The game is played by two players ($A$ and $B$) as
follows. Initially, an heap of $N$ balls is shared by the players, and
the urn is empty. Then, the players play by turn and pick either $1$
or $2$ balls from the heap and put them into the urn. A player looses
the game if he is the last to play (i.e., the heap is empty after he
has played). An arena modeling this game (for $N=8$) is given in
Fig.~\ref{fig:substrgame} (top), where $A$-states are circles,
$B$-states are squares, and the numbers labelling the states represent
the number of balls \emph{inside the urn}. The arena obtained from
Fig.~\ref{fig:substrgame} \emph{without the dotted edges} faithfully
models the description of the urn-filling game we have sketched above
(assuming Player~$A$ plays first). We say that a state is winning if,
from this state, Player $A$ has a \emph{winning strategy}, i.e. he can
always win the game whatever Player $B$ does (and vice-versa for
losing states). For instance, \badeight is loosing for
Player~$A$. Indeed, in this state, all $8$ balls are now into the urn,
hence the heap is empty. Moreover, \badeight belongs to Player~$B$,
hence Player~$A$ has necessarily played just before reaching it. State
\badseven can also be declared as losing, since Player~$A$ has no
other choice than tossing the last ball into the urn. Thus, in this
game, the objective of Player~$A$ is to avoid the set
$\bad=\left\{\badseven,\badeight\right\}$.  It is well-known
\cite{Bouton190235} that a simple characterisation of the set of
winning states\footnote{In order to make our example more interesting
  (this will become clear in the sequel), we have added the three
  \emph{dotted edges} from \Bstate{$7$} to \Astate{$6$} and
  \Astate{$5$} respectively, and from \Bgraystate{$6$} to \Astate{$5$}
  although those actions are not permitted in the original
  game. However, observe that those extra edges do not modify the set
  of winning states.} can be given. For each state $v$, let
$\lambda(v)$ denote its label. Then, the winning states (in white in
Fig~\ref{fig:substrgame}) are all the $A$-states $v$
s.t. $\lambda(v)\mod 3\neq 1$ plus all the $B$-states $v'$
s.t. $\lambda(v')\mod 3=1$.

\begin{figure}[t]
\begin{center}
\scalebox{0.75}{
\begin{tikzpicture}
  \tikzstyle{every state}=[text=black]
\everymath{\scriptstyle}

  \node[state, fill=white] (A8)  {\huge $0$};
  \node[state, fill=white] (A6) [right of=A8, node distance=2cm] {\huge $2$};
    \node[state, fill=white] (A5) [right of=A6, node distance=2cm] {\huge $3$};
  \node[state, fill=gray!20] (A4) [right of=A5, node distance=2cm] {\huge $4$};
  \node[state, fill=white] (A3) [right of=A4, node distance=2cm] {\huge $5$};
  \node[state, fill=white] (A2) [right of=A3, node distance=2cm] {\huge $6$};
  \node[state, fill=gray] (A1) [right of=A2, node distance=2cm] {\huge $7$};
  
  \node[state, rectangle, fill=white] (B7) [below of=A6, node distance=1.7cm,xshift=-2cm] {\huge $1$};
    \node[state, rectangle, fill=gray!20] (B6) [right of=B7, node distance=2cm] {\huge $2$};
      \node[state, rectangle, fill=gray!20] (B5) [right of=B6, node distance=2cm] {\huge $3$};
  \node[state, rectangle, fill=white] (B4) [right of=B5, node distance=2cm] {\huge$4$};
  \node[state, rectangle, fill=gray!20] (B3) [right of=B4, node distance=2cm] {\huge $5$};
  \node[state, rectangle, fill=gray!20] (B2) [right of=B3, node distance=2cm] {\huge $6$};
  \node[state, rectangle, fill=white] (B1) [right of=B2, node distance=2cm] {\huge $7$};
  \node[state, rectangle, fill=gray] (B0) [right of=B1, node distance=2cm] {\huge $8$};
  
       \draw[-latex'](-.7,0) -- (A8); 
       
\draw[-latex',line width=1.3pt](A8) -- (B7);% node [left,pos=.3] {$1$};
\draw[-latex'](A8) -- (B6);% node [left,pos=.3] {$2$};
\draw[-latex'](A6) -- (B5);% node [left,pos=.3] {$1$};
\draw[-latex',line width=1.3pt](A6) -- (B4);% node [left,pos=.3] {$2$};
\draw[-latex',line width=1.3pt](A5) -- (B4);% node [left,pos=.3] {$1$};
\draw[-latex'](A5) -- (B3);% node [left,pos=.3] {$2$};
\draw[-latex'](A4) -- (B3);% node [left,pos=.3] {$1$};
\draw[-latex'](A4) -- (B2);% node [left,pos=.3] {$2$};
\draw[-latex'](A3) -- (B2);% node [left,pos=.3] {$1$};
\draw[-latex',line width=2pt](A3) -- (B1);% node [left,pos=.3] {$2$};
\draw[-latex',line width=2pt](A2) -- (B1);%node [left,pos=.3] {$1$};
\draw[-latex'](A2) -- (B0) ;%node [left,pos=.3] {$2$};
\draw[-latex'](A1) edge (B0) ;%node [left,pos=.3] {$1$};
%\draw[-latex'](A1) edge[bend left=10] (B0);% node [left,pos=.3] {$2$};

\draw[-latex'](B7) -- (A6);% node [left,pos=.3] {$1$};
\draw[-latex'](B7) -- (A5);% node [left,pos=.3] {$2$};
\draw[-latex'](B6) -- (A5);% node [left,pos=.3] {$1$};
\draw[-latex'](B6) -- (A4);% node [left,pos=.3] {$2$};
\draw[-latex'](B5) -- (A4);% node [left,pos=.3] {$1$};
\draw[-latex'](B5) -- (A3);% node [left,pos=.3] {$2$};
\draw[-latex'](B4) -- (A3) ;%node [left,pos=.3] {$1$};
\draw[-latex'](B4) -- (A2);% node [left,pos=.3] {$2$};
\draw[-latex'](B3) -- (A2) ;%node [left,pos=.3] {$1$};
\draw[-latex'](B3) -- (A1);% node [left,pos=.3] {$2$};
\draw[-latex'](B2) -- (A1) ;%node [left,pos=.3] {$1$};
%\draw[-latex'](B2) -- (A8) ;%node [left,pos=.3] {$2$};
%\draw[-latex'](B1) edge[bend left=10] (A8);% node [left,pos=.3] {$1$};
%\draw[-latex'](B1) edge[bend left=20] (A8);% node [left,pos=.3] {$2$};
\draw[-latex'](B1) edge[dotted, bend left=12] (A2);% node [left,pos=.3] {$1$};
\draw[-latex'](B1) edge[dotted, bend left=12] (A3);% node [left,pos=.3] {$2$};
\draw[-latex'](B2) edge[dotted, bend left=12] (A3);% node [left,pos=.3] {$2$};

\end{tikzpicture}}

 \centering
\medskip  \begin{tabular}{c@{\hskip .5cm}c@{\hskip .5cm}c}
    (a) Winning strategy& (b) Winning $\star$-strategy & (c) $\wbs_0$-Winning $\star$-strategy\\
    \begin{tabular}{c|c@{\hskip .1cm}||@{\hskip .1cm}c|c}
      node&succ.&node&succ.\\\hline
      \Astate{0}&\Bstate{1}&\Astate{5}&\Bstate{7}\\
      \Astate{2}&\Bstate{4}&\Astate{6}&\Bstate{7}\\
      \Astate{3}&\Bstate{4}&\badseven&\badeight\\
      \Agraystate{4}&\Bstate{5}\\
    \end{tabular}
    &
    \begin{tabular}{c|c@{\hskip .1cm}||@{\hskip .1cm}c|c}
       node&succ.&node&succ.\\\hline
      \Astate{0}&\Bstate{1}&\Astate{5}&\Bstate{7}\\
      \Astate{2}&\Bstate{4}&\Astate{6}&\Bstate{7}\\
      \Astate{3}&\Bstate{4}\\
    \end{tabular}
    &
    \begin{tabular}{c|c}
       node&succ.\\\hline
      \Astate{5}&\Bstate{7}\\
      \Astate{6}&\Bstate{7}\\
    \end{tabular}
  \end{tabular}

\caption{Urn-filling Nim game with $N=8$, and three winning strategies.}\label{fig:substrgame}\label{fig:strategies}
\end{center}
  \vspace{-.7cm}
\end{figure}
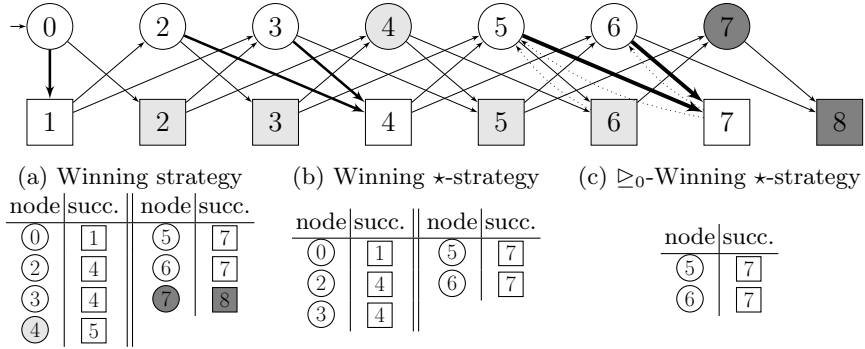

One of the nice properties of safety games is that they admit
\emph{memory-less winning strategies}, i.e., if Player $A$ has a
winning strategy in the game, then he has a winning strategy that
depends only on the current state (in other words, this strategy is a
function from the set of $A$-states to the set of
$B$-states). Memory-less strategies are often regarded as simple and
straightforward to implement (remember that the winning strategy is
very often the actual control policy that we want to implement in,
say, an embedded controller). Yet, this belief falls short in many
practical applications such as the three mentioned above because the
arena is not given explicitly, and its size is \emph{at least
  exponential} in the size of the original problem instance. Hence,
two difficulties arise, even with memory-less strategies. First, the
computation of winning strategies might request the traversal of the
whole arena, which might be intractable in practice. Second, a naive
implementation of a winning strategy $\sigma$ by means of an
exponentially large table, mapping each $A$-state $v$ to its safe
successor $\sigma(v)$, is not realistic. For example, a winning
strategy for our urn-filling game is given in
Fig.~\ref{fig:strategies}~(a), and it contains one line for each
$A$-state, even for the losing ones.

In this work, we consider the problem of computing \emph{efficiently}
winning strategies (for Player $A$) that can be \emph{succinctly}
represented. To formalise this problem, we consider
\emph{$\star$-strategies}, which are \emph{partial functions} defined
on a subset of $A$-states only. A $\star$-strategy can be regarded as
an \emph{abstract representation} of a family of (plain) strategies,
that we call \emph{concretisations} of the $\star$-strategies (they
are all the strategies that agree with the $\star$-strategy).  Then,
we want to compute \emph{succinct} $\star$-strategies that are defined
on the smallest possible number of states, but that are still safe
because all their concretisations are winning. In the example, a
winning $\star$-strategy of minimal size is given in
Fig.~\ref{fig:strategies}~(b). Remark that we can now omit two
$A$-states (the losing ones), because those states are not reachable
anyway, but we still need to remember what to play on all winning
states.  Unfortunately, we show (see
Section~\ref{sec:succinct-strategies}) that, unless P$=$NP, no
polynomial-time algorithm exists to compute minimal
$\star$-strategies, because the associated decision problem is
NP-complete. This holds when the arena is given explicitly, so the
difficulty is exacerbated in practice, where the arena is much larger
than the problem instance.

To cope with this complexity, we consider heuristics inspired from the
\emph{antichain} line of research \cite{DR10}. Antichain techniques
rely on a \emph{simulation partial order} on the states of the
system. This simulation is exploited to \emph{prune} the state space
that the algorithms need to explore, and to obtain \emph{efficient
  data structures} to store the set of states that the algorithms need
to maintain. These techniques have been applied to several relevant
problems, such as LTL model-checking \cite{DR10} or multi-processor
schedulability \cite{GGL-rts13} with remarkable performance
improvements of several orders of magnitude.

While a general theory of antichain has been developed \cite{DR10} in
the setting of \emph{automata-based models} (that are well-suited for
verification), they have been scarcely applied in the setting of
\emph{games}. One notable exception is the aforementioned work of
Filiot \textit{et al.} on LTL realisability \cite{FJR-fmsd11} where
the problem of realisability is reduced to a safety game, and a
simulation on the game states is exploited in an efficient algorithm
tailored for those games. The heuristics are thus limited to this
peculiar case, and their correctness cannot always be deduced from the
notion of \emph{simulation} but requires \emph{ad hoc} proofs.

In this paper, we advocate the use of a different notion,
i.e. \emph{turn-based alternating simulations} (tba-simulations for
short), which are adapted from the alternating simulations introduced
in \cite{AHKV-concur98}. Tba-simulations allow us to develop an
elegant and general theory, that extends the ideas of Filiot
\textit{et al.}  and that is applicable to a broad class of safety
games (including the three examples mentioned above). In our running
example, a tba-simulation $\wbs_0$ exists and is given by: $ v \wbs_0
v'$ iff $v$ and $v'$ belong to the same player, $\lambda(v)\geq
\lambda(v')$ and $\lambda(v)\mod 3= \lambda(v')\mod 3$. Then, it is
easy to see that the winning strategy of Fig.~\ref{fig:strategies}~(a)
exhibits some kind of \emph{monotonicity} wrt $\wbs_0$: for instance
$\Astate{5}\wbs_0\Astate{2}$, and the winning strategy consists in
putting two balls in the urn in both cases. Hence, we can further
reduce the representation of the strategy to the one in
Fig.~\ref{fig:strategies}~(c). Indeed, while not all concretisations
of this strategy are winning (for instance, the strategy does not say
what to play in \Astate{3} and obviously going to \Bgraystate{4}
should be avoided), all concretisations compatible with $\wbs$
\emph{are} winning. It is then easy to implement \emph{succinctly}
this strategy: only the table in Fig.~\ref{fig:strategies}~(c) needs
to be stored, because he relation $\wbs$ can be directly computed on
the description of the states. All these intuitions are formalised in
Section~\ref{sec:structured-games}, where we show that, in general, it
is sufficient to store the strategy on the maximal \emph{antichain} of
the reachable winning states, i.e., on a set of states that are all
incomparable by the tba-simulation, and thus very compact.

Next, we present, in Section~\ref{sec:effic-comp-succ} \emph{an
  efficient on-the-fly algorithm to compute such succinct
  $\star$-strategies}, which is adapted from the classical OTFUR
algorithm to solve safety games \cite{CDFLL-concur05}. Our algorithm
incorporates several heuristics that stem directly from the definition
of tba-simulation (in particular, it incorporates an optimisation that
was not present in the work of Filiot et
al. \cite{FJR-fmsd11}). Finally, in Section~\ref{sec:applications}, we
devise a criterion that allows to \emph{determine when a simulation
  relation on a game arena is also a tba-simulation}.  We show that
the safety games one considers in three applications introduced above
(LTL realisability, real-time feasibility and determinisation of timed
automata) respect this criterion, which demonstrates the wide
applicability of our approach.

\section{Preliminaries}
\paragraph{Turn-based safety games} A \emph{finite turn-based game
  arena} is a tuple $\Aa=(V_A,V_B,E,I)$, where $V_A$ and $V_B$ are the
finite sets of states controlled by Players~$A$ and $B$ respectively;
$E\subseteq (V_A\times V_B) \cup (V_B\times V_A)$ is the set of edges;
and $I\in V_A$ is the initial state. We denote by $V$ the set $V_A\cup
V_B$. Given a finite arena $\Aa=(V_A,V_B,E,I)$ and a state $v\in V$,
we let $\Succ{\Aa,v}=\{v'\mid (v,v')\in E\}$ and
$\Reach{\Aa,v}=\{v'\mid (v,v')\in E^*\}$, where $E^*$ denotes the
reflexive and transitive closure of $E$. We lift the definitions of
$\Reachname$ and $\Succname$ to sets of states in the usual way.
A \emph{finite turn-based safety game} is a tuple
$G=(V_A,V_B,E,I,\bad)$ where $(V_A, V_B, E, I)$ is a finite turn-based
game arena, and $\bad\subseteq V$ is the set of bad states that $A$
wants to avoid. The definitions of $\Reachname$ and $\Succname$ carry
on to games: $\Reach{(\Aa,\bad),v}=\Reach{\Aa,v}$ and
$\Succ{(\Aa,\bad),v}=\Succ{\Aa,v}$. When the game is clear from the
context, we often omit it.

\paragraph{Plays and strategies} During the game, players interact to
produce a play, which is a finite or infinite path in the graph
$(V,E)$. Players play turn by turn, by moving a \emph{token} on the
game's states. Initially, the token is on state $I$. At each turn, the
player who controls the state marked by the token gets to choose the
next state.
A \emph{strategy} for $A$ is a function $\sigma:V_A \rightarrow V_B$
such that for all $v\in V_A$, $(v,\sigma(v))\in E$.  We extend
strategies to set of states $S$ in the usual way: $\sigma(S)=
\{\sigma(v)\mid v\in S\}$. A strategy $\sigma$ for $A$ is
\emph{winning for a state $v\in V$} iff no bad states are reachable
from $v$ in the graph $G_\sigma$ obtained from $G$ by removing all the
moves of $A$ which are not chosen by $\sigma$,
i.e. $\Reach{G_\sigma,v}\cap\bad=\emptyset$, where $G_\sigma=(V_A,
V_B, E_\sigma, I,\bad)$ and $E_\sigma=\{(v,v')\mid (v,v')\in E\wedge
v\in V_A\implies v'=\sigma(v)\}$. We say that a strategy $\sigma$ is
\emph{winning} in a game $G=(V_A,V_B,E,I,\bad)$ iff it is winning in
$G$ for $I$.

\paragraph{Winning states and attractors} A state $v\in V$ in $G$ is
\emph{winning} iff there exists a strategy $\sigma$ that is winning in
$G$ for~$v$. We denote by $\win$ the set of winning states (for
Player~$A$). By definition, any strategy such that
$\sigma(\win)\subseteq \win$ is thus winning.  Moreover, it is
well-known that the set $\win$ of winning states can be computed in
polynomial time (in the size of the arena), by computing the so-called
\emph{attractor} of the unsafe states. In a game $G=(V_A, V_B, E,
I,\bad)$, the sequence $(\Attr_i)_{i\geq 0}$ of attractors (of the
$\bad$ states) is defined as follows. $\Attr_0=\bad$ and for all $i\in
\mathbb{N}$, $\Attr_{i+1}= \Attr_i \cup \{v\in V_B \mid
\Succ{v}\cap\Attr_i\neq\emptyset\}\cup\{v\in V_A\mid
\Succ{v}\subseteq\Attr_i\}$.  It is well-known that, for all $i\geq
0$, $v\in\Attr_i$ means that, from $v$, Player $B$ can force the game
to visit $\bad$ in at most $i$ steps. For finite games, the sequence
stabilises after a finite number of steps on a set of states that we
denote $\Attr_\bad$. Then, $v$ belongs to $\Attr_\bad$ iff Player $B$
can force the game to visit $\bad$ from $v$. Thus, the set of winning
states for player A is $\win= V\setminus \Attr_\bad$.  It is then
straightforward to compute a winning strategy $\sigma$: for all $v\in
V_A\cap\win$, we let $\sigma(v)=v'$, where $v'\in\win$. Such a node is
guaranteed to exist by definition of the attractor.

\paragraph{Partial orders, closed sets and antichains}
Fix a finite set $S$. A relation $\wbs\in S\times S$ is a partial
order iff $\wbs$ is reflexive, transitive and antisymmetric, i.e. for
all $s\in S$: $(s,s)\in \wbs$ (reflexivity); for all $s,s',s''\in S$,
$(s,s')\in \wbs$ and $(s',s'')\in \wbs$ implies $(s,s'')\in \wbs$
(transitivity); and for all $s,s'\in S$: $(s,s')\in \wbs$ and
$(s',s)\in \wbs$ implies $s=s'$ (antisymmetry).  As usual, we often
write $s\wbs s'$ and $s\nwbs s'$ instead of $(s,s')\in \wbs$ and
$(s,s')\not\in \wbs$, respectively. The $\wbs$-\emph{downward closure}
$\dco{S'}{\wbs}$ of a set $S'\subseteq S$ is defined as
$\dco{S'}{\wbs}= \{s \mid \exists s'\in S', s'\wbs
s\}$. Symmetrically, the \emph{upward closure} $\uco{S'}{\wbs}$ of
$S'$ is defined as: $\uco{S'}{\wbs}=\{s \mid \exists s'\in S': s\wbs
s'\}$. Then, a set $S'$ is \emph{downward closed} (resp. \emph{upward
  closed}) iff $S' = \dco{S'}{\wbs}$ (resp. $S' =
\uco{S'}{\wbs}$). When the partial order is clear from the context, we
often write $\dc{S}$ and $\uc{S}$ instead of $\dco{S}{\wbs}$ and
$\uco{S}{\wbs}$ respectively. Finally, a subset $\alpha$ of some set
$S'\subseteq S$ is an \emph{antichain} on $S'$ with respect to $\wbs$
if for all $s, s' \in \alpha$, $s\nwbs s'$.  An antichain $\alpha$ on
$S'$ is said to be a set of \emph{maximal elements of $S'$} (or,
simply \emph{a maximal antichain} of $S'$) iff for all $s_1\in S'$
there is $s_2 \in \alpha$: $s_2\wbs s_1$. Symmetrically, an antichain
$\alpha$ on $S'$ is a set of \emph{minimal elements of $S'$} (or
\emph{a minimal antichain} of $S'$) iff for all $s_1\in S'$ there is
$s_2 \in \alpha$: $s_1\wbs s_2$. It is easy to check that if $\alpha$
and $\beta$ are maximal and minimal antichains of $S'$ respectively,
then $\dc{\alpha}=\dc{S'}$ and $\uc{\beta}=\uc{S'}$. Intuitively,
$\alpha$ ($\beta$) can be regarded as a symbolic representation of
$\dc{S'}$ ($\uc{S'}$), which is of minimal size in the sense that it
contains no pair of $\wbs$-comparable elements. Moreover, since $\wbs$
is a partial order, each subset $S'$ of the finite set $S$ admits a
unique minimal and a unique maximal antichain, that we denote by
$\minac{S'}$ and $\maxac{S'}$ respectively. Observe that one can
always effectively build a $\maxac{S'}$ and $\minac{S'}$, simply by
iteratively removing from $S'$, all the elements that are strictly
$\wbs$-dominated by (for $\maxac{S'}$) or that strictly dominate (for
$\minac{S'}$) another one.

\paragraph{Simulation relations} Fix an arena
$G=(V_A,V_B,E,I,\bad)$. A relation $\wbs\subseteq V_A\times V_A \cup
V_B\times V_B$ is a \emph{simulation relation compatible\footnote{See
    \cite{FJR-fmsd11} for an earlier definition of a simulation
    relation compatible with a set of states.} with $\bad$} (or simply
a \emph{simulation}) iff it is a partial order\footnote{Observe that
  our results can be extended to the case where the relations are
  \emph{preorders}, i.e. transitive and reflexive relations.} and for
all $(v_1,v_2)\in \wbs$: either $v_1\in \bad$ or:
\begin{inparaenum}[(i)]
\item for all $v_2'\in\Succ{v_2}$, there is $v_1'\in\Succ{v_1}$
  s.t. $v_1'\wbs v_2'$ and
\item \label{item:3}$v_2 \in \mathsf{Bad}$ implies that $v_1 \in
  \mathsf{Bad}$.
\end{inparaenum}
On our example, the relation $\wbs_0 =\{(v,v')\in V_A\times V_A \cup
V_B \times V_B\mid \lambda(v)\geq \lambda(v')\textrm{ and }
\lambda(v)\mod 3= \lambda(v')\mod 3\}$ is a simulation relation
compatible with $\bad=\left\{\badseven,\badeight\right\}$.  Moreover,
$\win=\{v\in V_A\mid \lambda(v)\mod 3\neq 1\}\cup\{v\in V_B\mid
\lambda(v)\mod 3= 1\}$ is downward closed for $\wbs_0$ and its
complement (the set of losing states), is upward closed.  Finally,
$\win$ admits a single maximal antichain for $\wbs_0$:
$\maxwin=\left\{\Bstate{7}, \Astate{6}, \Astate{5}\right\}$.

\section{Succinct strategies\label{sec:succinct-strategies}}
Let us first formalise the notion of \emph{succinct strategy}. As
explained in the introduction, a naive way to implement a memory-less
strategy $\sigma$ is to store, in an appropriate data structure, the
set of pairs $\{(v, \sigma(v))\mid v\in V_A\}$, and implement a
controller that traverses the whole table to find action to perform
each time the system state is updated. While the definition of
strategy asks that $\sigma(v)$ be defined for all $A$-states $v$, this
information is sometimes indifferent, for instance, when $v$ is not
reachable in $G_\sigma$. Thus, we want to reduce the number of states
$v$ s.t. $\sigma(v)$ is crucial to keep the system safe.

\paragraph{$\star$-strategies} We introduce the notion of
$\star$-strategy to formalise this idea: a $\star$-strategy is a
function $\hat{\sigma}:V_A\mapsto V_B\cup\{\star\}$, where $\star$
stands for a `don't care' information. We denote by
$\cdom{\hat{\sigma}}$ the \emph{support} $\hat{\sigma}^{-1}(V_B)$ of
$\hat{\sigma}$, i.e. the set of nodes $v$
s.t. $\hat{\sigma}(v)\neq\star$.  Such $\star$-strategies can be
regarded as a representation of a family of concrete strategies. A
\emph{concretisation} $\sigma$ of a $\star$-strategy $\hat{\sigma}$ is
a strategy $\sigma$ s.t. for all $v\in V_A$, $\hat{\sigma}(v) \neq
\star$ implies $\hat{\sigma}(v) = \sigma(v)$. A $\star$-strategy
$\hat{\sigma}$ is \emph{winning} if every concretisation of
$\hat{\sigma}$ is winning (intuitively, $\hat{\sigma}$ is winning if
$A$ always wins when he plays according to $\hat{\sigma}$, whatever
choices he makes when $\hat{\sigma}$ returns $\star$).  The
\emph{size} of a $\star$-strategy $\hat{\sigma}(v)$ is the size of 
 $\cdom{\hat{\sigma}}$. On our running example, the strategy
$\sigma$ displayed in Fig.~\ref{fig:strategies}~(b) -- assuming the
lines where $\sigma(v)=\star$ have been omitted -- is a winning
$\star$-strategy of minimal size.

\paragraph{Computing succinct $\star$-strategies} Our goal is to
compute \emph{succinct} $\star$-strategies, defined as
$\star$-strategies of minimal size. To characterise the hardness of
this task, we consider the following decision problem, and prove that
it is NP-complete:
\begin{problem}[\mss]
  Given a finite turn-based safety game $G$ and an integer
  $k\in\mathbb{N}$, decide whether there is a winning $\star$-strategy
  of size smaller than $k$ in $G$.
\end{problem}
\begin{theorem}\label{the:mss-np-complete}
  \mss is NP-complete.
\end{theorem}
\begin{proof}
  Let $X=\{x_1,\cdots,x_m\}$ be a set of $m$ Boolean variables. Let
  $\phi = \land_{1\le i\le n}C_i$ be a \textsc{Sat} formula (with $n$
  clauses) such that for all $1\le i\le n$, $C_i = \lor_{1 \le j \le
    n_i} l^i_j$ with $l^i_j \in L=\{x_1,\cdots, x_m, \lnot x_1,\cdots,
  \lnot x_m\}$ ($L$ is called the set of literals). Then, let us
  build a finite turn-based safety game $G$ and an integer $k$ in
  polynomial time as follows.  We let $G=(V_A,V_B,E,I,\mathsf{Bad})$
  where:
  \begin{itemize}
  \item $V_A = \{\text{init}^A\}\cup\LL^A\cup X\cup C$, where
    $\LL^A=\{x^A_1,\cdots, x^A_m, \lnot x^A_1,\cdots, \lnot x^A_m\}$,
    $\XX=\{X_1,\cdots, X_m\}$ and $\CC=\{C_1,\cdots,C_n\}$
  \item $V_B = \{\text{init}^B,\text{bad}\}\cup \{x^B_1,\cdots, x^B_m,
    \lnot x^B_1,\cdots, \lnot x^B_m\}$
  \item $E =
    \{(\text{init}^A,\text{init}^B)\} %
    \cup\{\text{init}^B\}\times(\XX\cup \CC)\\ %
    \cup \left\{
      \begin{array}{c|l}
        (X_i,x^B_i), (X_i,\lnot x^B_i),(x^B_i,x^A_i),(\lnot x^B_i,\lnot x^A_i),\\
        (x^A_i,\text{init}^B), (\lnot x^A_i,\text{init}^B),(x^A_i,\text{bad}), 
        &
          1\le i \le m\\
          (\lnot x^A_i,\text{bad}),(X_i,\text{bad}),(C_i,\text{bad})
      \end{array}
      \right\}\\
    \cup (\LL^A\cup \CC\cup \XX)\times \{\text{bad}\}\\ %
    \cup \{(C_i,x^B_h)\mid 1\le i \le n, \exists 1 \le j \le n_i:  l^i_j=x_h\}\\ %
    \cup\{(C_i,\lnot x^B_h) \; |\;1\le i \le n, \exists 1 \le j \le n_i: l^i_j=\lnot x_h\}\\ %
    $
  \item $I = \text{init}^A$
  \item $\mathsf{Bad} = \{\text{bad}\}$
  \end{itemize}
  Finally, let $k=2\times m + n$. An example of the construction for
  the formula $\phi=(x_1\lor x_2 \lor \lnot x_3)\wedge (\lnot x_1\lor
  x_2 \lor x_3) $ is given in Figure~\ref{fig:redex} (note that
  $\text{bad}$ and the $(\lnot)x_i^B$'s) have been duplicated to
  enhance readability). States of $A$ are circles and states of $B$
  are squares.

  Let us show that $\phi$ is satisfiable iff there is winning
  $\star$-strategy of size at most $k$ in $G$.  To this end, we first
  make several observations on $G$. First, there is a winning strategy
  in $G$ since all predecessors of $bad$ are Player $A$ states that
  have other successors that $\text{bad}$. Second, Player $A$ can
  never avoid the states of the form $X_i$ nor $C_j$, i.e. for all
  strategy $\sigma$: $\XX\cup \CC
  \subseteq\Reach{G_\sigma,\text{init}^A}$. This entails that, in all
  \emph{winning} $\star$-strategy $\hat{\sigma}$, we must have
  $\hat{\sigma}(v)\neq\star$ for all $v\in \XX\cup \CC$. Otherwise, if
  $\hat{\sigma}(v)=\star$ for some $v\in \XX\cup \CC$, there is at
  least one concretisation $\sigma$ of $\hat{\sigma}$
  s.t. $\sigma(v)=\text{bad}$, and thus $\text{bad}$ is reachable in
  $G_\sigma$ by the path $\text{init}^A, \text{init}^B,v,\text{bad}$,
  which contradicts the fact that $\hat{\sigma}$ is winning. Finally,
  consider a winning $\star$-strategy $\hat{\sigma}$. We have shown
  above that $\hat{\sigma}(X_i)\not\in\{\star,\text{bad}\}$ for all
  $1\leq i\leq m$. This implies that $\hat{\sigma}$ maps each $X_i$
  either to $x_i^B$, or to $\lnot x_i^B$, and those nodes cannot be
  avoided by Player $A$, whatever concretisation of $\hat{\sigma}$ he
  plays. Using the same arguments as above, we conclude that
  $\hat{\sigma}(v)$ must be different from $\star$ for exactly $m$
  states in $\LL^A$. More precisely, for all winning
  $\star$-strategies $\hat{\sigma}$, let $S_{\hat{\sigma}}=\{x_i^A\mid
  \exists 1\leq i\leq m: \hat{\sigma}(X_i)=x_i^B\}\cup\{\lnot
  x_i^A\mid \exists 1\leq i\leq m: \hat{\sigma}(X_i)=\lnot
  x_i^B\}$. Then, for all winning $\star$-strategies $\hat{\sigma}$,
  for all $v\in S_{\hat{\sigma}}$: $\hat{\sigma}(v)\neq\star$. Observe
  that $|S_{\hat{\sigma}}|=m$ for all winning $\star$-strategies
  $\hat{\sigma}$. We conclude that all winning $\star$-strategies
  $\hat{\sigma}$ is of size at least $k=2\times m +n$ in $G$.

  Let us now conclude the proof. First consider a satisfying
  assignment $a:X\mapsto \{\mathbf{true},\mathbf{false}\}$ for
  $\phi$. Let $\hat{\sigma}$ be the $\star$-strategy defined as
  follows. For all $1\leq i\leq m$: $\hat{\sigma}(X_i)=x^B_i$ if
  $a(x_i)=\mathbf{true}$ and $\hat{\sigma}(X_i)=\lnot x^B_i$
  otherwise. For all $1\leq i\leq n$: pick a literal $\ell_i$ from
  $C_i$ which is true under the assignment $a$ (such a literal must
  exist since $a$ makes $\phi$ true), and let
  $\hat{\sigma}(C_i)=\ell_i^B$. For all $v\in \{x_1^A, \lnot
  x_1^A,\ldots, x_n^A, \lnot x_n^A\}$, let
  $\hat{\sigma}(v)=\text{init}^B$ if $v=\hat{\sigma}(X_i)$ for some
  $1\leq i\leq n$, and let $\hat{\sigma}(v)=\star$ otherwise. Finally,
  let $\hat{\sigma}(\text{init}^A)=\star$. It is easy to check that
  $\hat{\sigma}$ is indeed a \emph{winning} $\star$-strategy of size
  exactly $k$.

  Second, we assume a  winning $\star$-strategy $\overline{\sigma}$
  of size $\leq k$. By the above arguments, $\overline{\sigma}$ is of
  size exactly $k$ and $\overline{\sigma}(v)\neq\star$ for all $v\in
  \XX\cup \CC\cup S_{\overline{\sigma}}$. Let us consider the assignment
  of the variables of $\phi$ that maps each variable $x_i$ to
  \textbf{true} iff $\overline{\sigma}(X_i)=x^B_i$. To show that this
  assignment satisfies $\phi$ it is sufficient to show that
  $\overline{\sigma}(\CC) \subseteq \overline{\sigma}(\XX)$, because this
  entails, by definition of $G$, that, under this assignment, each
  clause $j$ contains at least one true literal (the one corresponding to
  $\overline{\sigma}(C_j)$).  To establish that $\overline{\sigma}(\CC)
  \subseteq \overline{\sigma}(\XX)$, we proceed by contradiction. If it
  is not the case, let $v\in \CC$ be a state
  s.t. $\overline{\sigma}(v)\not\in\overline{\sigma}(\XX)$. Assume
  $\overline{\sigma}(v)=\sim x^B_k$ for some $k$, and where $\sim$ can
  be $\lnot$ or nothing. Then, by definition of $G$, the corresponding
  $A$ state $\sim x^A_k$ is reachable in $G_\sigma$ for all
  concretisation $\sigma$ of $\overline{\sigma}$. Moreover, $\sim
  x^A_k\not\in X\cup C\cup S_{\overline{\sigma}}$, hence
  $\overline{\sigma}(\sim x^A_k)=\star$, and there is at least one
  concretisation of $\overline{\sigma}$ that maps $\sim x^A_k$ to
  $\text{bad}$. Since $\sim x^A_k$ is reachable in this concretisation
  in particular, we conclude that $\overline{\sigma}$ is not
  winning. Contradiction. \qed
\end{proof}

\begin{figure}[t]
\begin{center}
\scalebox{0.6}{
\begin{tikzpicture}[->,>=stealth',shorten >=1pt,auto,node distance=2cm,
                    semithick]

  \tikzstyle{every state}=[text=black]

  \node[state, rectangle, fill=white] (A)  {init$^B$};
  \node[state, fill=white] (A0) [above of=A, node distance=2cm] {init$^A$};
  \node[state, fill=white] (X1) [right of=A, node distance=4cm, yshift=3cm] {$X_1$};
  \node[state, fill=white] (X2) [right of=A, node distance=8cm, yshift=3cm] {$X_2$};
  \node[state] (X3) [right of=A, node distance=12cm, yshift=3cm] {$X_3$};
  
  \node[state, rectangle, fill=gray!30] (bad) [above of=X2, node distance=3cm, xshift=0cm] {bad};
    \node[state, rectangle, fill=gray!30] (bad') [above of=X2, node distance=3cm, yshift=-10cm,xshift=-2cm] {bad};

  \node[state, rectangle, fill=white] (x2) [below of=X2, node distance=1.5cm, xshift=-1cm] {$x^B_2$}; 
  \node[state, rectangle, fill=white] (nx2) [below of=X2, node distance=1.5cm, xshift=1cm] {$\lnot x^B_2$};
    \node[state, rectangle, fill=white] (x2') [below of=X2, node distance=4.5cm, xshift=-1cm] {$x^B_2$}; 
  \node[state, rectangle, fill=white] (nx2') [below of=X2, node distance=4.5cm, xshift=1cm] {$\lnot x^B_2$};
   \node[state, fill=white] (x2A) [below of=x2, node distance=1.5cm, xshift=-1cm] {$x^A_2$}; 
  \node[state, fill=white] (nx2A) [below of=nx2, node distance=1.5cm, xshift=-1cm] {$\lnot x^A_2$};
  \node[state, rectangle, fill=white] (x1) [below of=X1, node distance=1.5cm, xshift=-1cm] {$x^B_1$}; 
  \node[state, rectangle, fill=white] (nx1) [below of=X1, node distance=1.5cm, xshift=1cm] {$\lnot x^B_1$};
    \node[state, rectangle, fill=white] (x1') [below of=X1, node distance=4.5cm, xshift=-1cm] {$x^B_1$}; 
  \node[state, rectangle, fill=white] (nx1') [below of=X1, node distance=4.5cm, xshift=1cm] {$\lnot x^B_1$};
   \node[state, fill=white] (x1A) [below of=x1, node distance=1.5cm, xshift=-1cm] {$x^A_1$}; 
  \node[state, fill=white] (nx1A) [below of=nx1, node distance=1.5cm, xshift=-1cm] {$\lnot x^A_1$};
    \node[state, rectangle, fill=white] (x3) [below of=X3, node distance=1.5cm, xshift=-1cm] {$x^B_3$}; 
  \node[state, rectangle, fill=white] (nx3) [below of=X3, node distance=1.5cm, xshift=1cm] {$\lnot x^B_3$};    \node[state, rectangle, fill=white] (x3') [below of=X3, node distance=4.5cm, xshift=-1cm] {$x^B_3$}; 
  \node[state, rectangle, fill=white] (nx3') [below of=X3, node distance=4.5cm, xshift=1cm] {$\lnot x^B_3$};
   \node[state, fill=white] (x3A) [below of=x3, node distance=1.5cm, xshift=-1cm] {$x^A_3$}; 
  \node[state, fill=white] (nx3A) [below of=nx3, node distance=1.5cm, xshift=-1cm] {$\lnot x^A_3$};

  \node[state, fill=white] (C1) [right of=A, node distance=4cm, yshift=-4cm] {$C_1$};
  \node[state] (C2) [right of=A, node distance=8cm, yshift=-4cm] {$C_2$};

 \path (A) edge[bend left] (X1)
        (A) edge[bend left=60]  (X2)
        (A) edge[bend left=65]  (X3)
        (A) edge[bend right] (C1)
        (x2) edge (x2A)
        (x2') edge (x2A)
        (x1) edge (x1A)
        (x1') edge (x1A)
        (x3) edge (x3A)
        (x3') edge (x3A)
        (nx2') edge (nx2A)
        (nx1) edge (nx1A)
        (nx1') edge (nx1A)
        (nx3) edge (nx3A)
        (nx3') edge (nx3A)
        (0,3.5) edge (A0)
        (nx2) edge (nx2A)
        (A) edge [bend right=55](C2);
     \path    (X1) [dashed] edge (bad)
        (A0) edge (A)
     (X2) edge (bad)
     (X3) edge[bend right=15] (bad)
     (x2A) edge[bend left=15] (bad)
     (nx2A) edge[bend right=15] (bad)
     (x2A) edge[bend right=22] (A)
     (nx2A) edge[bend right=25] (A)
     (x1A) edge[bend left=-10] (bad)
     (nx1A) edge[bend right=25] (bad)
     (x1A) edge[bend right=15] (A)
     (nx1A) edge[bend right=25] (A)
     (x3A) edge[bend left=-15] (bad)
     (nx3A) edge[bend right=25] (bad)
     (x3A) edge[bend right=-17] (A)
     (nx3A) edge[bend right=-16] (A)
     (X2) edge (x2)
     (X2) edge (nx2)
     (X1) edge (x1)
     (X1) edge (nx1)
     (X3) edge (x3)
     (X3) edge (nx3)
     (C1) edge (x2')
     (C2) edge (x2')
     (C1) edge (x1')
     (C1) edge (nx3')
     (C2) edge (nx1')
     (C2) edge (x3')
        (C1) edge[bend left=-15]  (bad')
        (C2) edge[bend right=-15]  (bad');
\end{tikzpicture}}
\vspace{-.5cm}
\caption{Construction for the formula $\varphi=(x_1\lor x_2 \lor \lnot
  x_3)\wedge (\lnot x_1\lor x_2 \lor x_3) $: $k = 8$. State $bad$, and
  all the states of the form $(\neg)x^B_i$ have been duplicated for
  readability.  }\label{fig:redex}
\end{center}
  \vspace{-.7cm}
\end{figure}
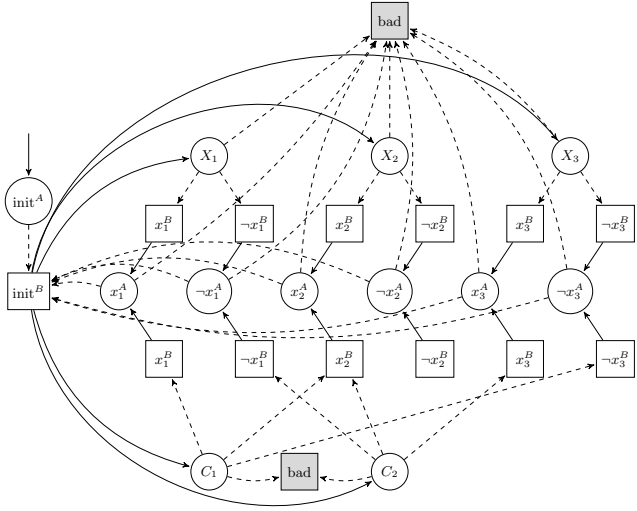

Thus, unless P=NP, there is no polynomial-time algorithm to compute a
winning $\star$-strategy \emph{of minimal size}.
In most practical cases we are aware of, the situation is even worse,
since the arena is not given explicitly. This is the case with the
three problems we consider as applications (see
Section~\ref{sec:applications}), because they can be reduced to a
safety games whose size is at least \emph{exponential} in the size of
the original problem instance.

\section{Structured games and monotonic strategies\label{sec:structured-games}}

To mitigate the strong complexity identified in the previous section,
we propose to follow the successful \emph{antichain approach}
\cite{WDMR-tacas08,DR10,FJR-fmsd11}. In this line of research, the
authors point out the fact that, in practical applications (like those
we identify in Section~\ref{sec:applications}), system states exhibit
some inherent \emph{structure}, which is formalised by a
\emph{simulation relation} and can be exploited to improve the
practical running time of the algorithms. In the present paper, we
rely on the notion of \emph{turn-based alternating simulation}, to
define heuristics to
\begin{inparaenum}[(i)]
\item improve the running time of the algorithms to solve safety games
  and
\item obtain succinct representations of strategies.
\end{inparaenum}
This notion is a straightforward adaptation, to turn-based games, of
the alternating simulations from \cite{AHKV-concur98}.

\paragraph{Turn-based alternating simulations}
Let $G=(V_A,V_B,E,I,\bad)$ be a finite safety game. A partial order
$\wbs\subseteq V_A\times V_A \cup V_B\times V_B$ is a \emph{turn-based
  alternating simulation relation for $G$} \cite{AHKV-concur98}
(tba-simulation for short) iff for all $v_1$, $v_2$ s.t. $v_1\wbs
v_2$, either $v_1\in \bad$ or the three following conditions hold:
\begin{inparaenum}[(i)]
\item If $v_1\in V_A$, then, for all
  $v_1'\in\Succ{v_1}$, there is $v_2'\in\Succ{v_2}$ s.t.  $v'_1\wbs
  v'_2$; 
\item If $v_1\in V_B$, then, for all $v_2'\in\Succ{v_2}$, there is
  $v'_1\in\Succ{v_1}$ s.t. $v'_1\wbs v'_2$; and
\item $v_2 \in \bad$ implies $v_1 \in \bad$.
\end{inparaenum}

On the running example (Fig.~\ref{fig:substrgame}), $\wbs_0$ is a
tba-simulation relation. Indeed, as we are going to see in
Section~\ref{sec:applications}, a simulation relation in a game where
player $A$ has always the opportunity to perform the same moves is
necessarily alternating.

\paragraph{Monotonic concretisations of $\star$-strategies} Let us
exploit the notion of tba-simulation to introduce a finer notion of
concretisation of $\star$-strategies.  Let $\hat{\sigma}$ be a
$\star$-strategy. Then, a strategy $\sigma$ is a
\emph{$\wbs$-concretisation} of $\hat{\sigma}$ iff for all $v\in V_A$:
\begin{align*}
  v\in\cdom{\hat{\sigma}} &\implies \sigma(v)=\hat{\sigma}(v)\\
  &\textrm{and}\\
  \left(v\not\in \cdom{\hat{\sigma}} \land v\in \dco{\cdom{\hat{\sigma}}}{\wbs}\right) &\implies
  \exists \overline{v}\in \cdom{\hat{\sigma}}: \overline{v}\wbs v\land\sigma(\overline{v})\wbs  \sigma(v)
\end{align*}
Intuitively, when $\hat{\sigma}(v)=\star$, but there is $v'\wbs v$
s.t. $\hat{\sigma}(v')\neq\star$, then, $\sigma(v)$ must mimic the
strategy $\sigma(\overline{v})$ from some state $\overline{v}$ that
covers $v$ and s.t. $\hat{\sigma}(\overline{v})\neq\star$. Then, we
say that a $\star$-strategy is \emph{$\wbs$-winning} if all its
$\wbs$-concretisations are winning.

Because equality is a tba-simulation, the proof of
Theorem~\ref{the:mss-np-complete} can be used to show that computing a
$\wbs$-winning $\star$-strategy of size less than $k$ is an
NP-complete problem too. Nevertheless, $\wbs$-winning $\star$-
strategies can be even more compact than winning $\star$-strategy. For
instance, on the running example, the smallest winning
$\star$-strategy $\overline{\sigma}$ is of size~$5$: it is given in
Fig.~\ref{fig:strategies}~(b) and highlighted by bold arrows in
Fig.~\ref{fig:substrgame} (thus,
$\overline{\sigma}(\Agraystate{4})=\overline{\sigma}(\badseven)=\star$). Yet,
one can define a $\wbs_0$-winning $\star$-strategy $\hat{\sigma}$ of
size~$2$ because states $\Astate{5}$ and $\Astate{6}$ simulate all the
winning states of $A$. This $\star$-strategy\footnote{Actually, this
  strategy is winning for all initial number $n$ of balls s.t. $n\mod
  3\neq 1$.} $\hat{\sigma}$ is the one given in
Fig.~\ref{fig:strategies}~(c) and represented by the boldest arrows in
Fig.~\ref{fig:substrgame}. Observe that, while all
$\wbs$-concretisations of $\hat{\sigma}$ are winning, it is not the
case of all \emph{concretisations} of $\hat{\sigma}$. For instance,
there is one concretisation $\sigma$ of $\hat{\sigma}$
s.t. $\sigma(\Astate{0})=\Bgraystate{2}$, but $\sigma$ is not a
$\wbs_0$-monotonic concretisation of $\hat{\sigma}$ (and is losing).

\paragraph{Obtaining $\wbs$-winning $\star$-strategies} The previous
example clearly shows the kind of $\wbs$-winning $\star$-strategies we
want to achieve: $\star$-strategies $\hat{\sigma}$
s.t. $\cdom{\hat{\sigma}}$ is the maximal antichain of the winning
states. In Section~\ref{algo:otfur_plusplus}, we introduce an
efficient on-the-fly algorithm to compute such a $\star$-strategy. Its
correctness is based on the fact that we can extract a $\wbs$-winning
$\star$-strategy from any winning (plain) strategy, as shown by
Proposition~\ref{th-monog} hereunder. For all strategy $\sigma$, and
all $V\subseteq V_A$, we let $\sigma|_V$ denote the $\star$-strategy
$\hat{\sigma}$ s.t. $\hat{\sigma}(v)=\sigma(v)$ for all $v\in V$ and
$\hat{\sigma}(v)=\star$ for all $v\not\in V$. Then:
\begin{proposition}\label{th-monog}
  Let $G=(V_A,V_B,E,I,\bad)$ be a finite turn-based safety game and
  $\wbs$ be a tba-simulation relation for $G$. Let $\sigma$ be a
  strategy in $G$, and let $\SS\subseteq V_A$ be a set of $A$-states
  s.t.:
  \begin{inparaenum}[(i)]
  \item \label{item:hyp1} $(\SS\cup\sigma(\SS))\cap \bad = \emptyset$;
  \item \label{item:hyp2} $I \in \dco{\SS}{\wbs}$; and
  \item \label{item:hyp3} $\mathsf{succ}(\sigma(\SS))\subseteq \dco{\SS}{\wbs}$.
  \end{inparaenum}
  Then, $\sigma|_\SS$ is a $\wbs$-winning $\star$-strategy.
\end{proposition}
\begin{proof}
  Let $\tau$ be a $\wbs$-concretisation of $\sigma|_{\SS}$ and let us
  show that $\tau$ is winning. Let us first show that all $A$-states
  reachable in $G$ under strategy $\tau$ are covered by some state in
  $\SS$, i.e. that $\Reach{G_{\tau}}\cap V_A \subseteq \dc{\SS}$. Let
  us consider $v\in \Reach{G_{\tau}}\cap V_A$, and let $v^A_0,v^B_1,
  v^A_1,\ldots, v^B_n, v^A_n$ be a path in in $G_\tau$ that reaches
  $v$, i.e., with $v^A_0=I$ and $v^A_n=v$. Let us prove, by
  induction on $i$ that $v^A_i\in\dc{\SS}$ for all $0\leq i\leq n$.

  \noindent{\bf Base case $i=0$}: trivial by
  hypothesis~(\ref{item:hyp2}).

  \noindent{\bf Inductive case $i=k>0$}: let us assume that
  $v^A_{k-1}\in\dc{\SS}$ and let us show that
  $v^A_k\in\dc{\SS}$. Since $(v^A_{k-1}, v^B_k)$ is an edge in
  $G_\tau$, $v^B_k=\tau(v^A_{k-1})$. Since $\tau$ is a
  $\wbs$-concretisation of $\sigma|_\SS$, there is $\overline{v}\in
  \SS$ s.t. $\overline{v}\wbs v^A_{k-1}$ and
  $\sigma(\overline{v})=\tau(\overline{v})\wbs \tau(v^A_{k-1})=
  v^B_k$. Thus, $\sigma(\overline{v})\wbs v^B_k$. Since $\wbs$ is a
  tba-simulation, the successor $v^A_k$ of $v^B_k$ can be simulated by
  some successor $\hat{v}$ of $\sigma(\overline{v})$,
  i.e. $\hat{v}\wbs v^A_k$. By hypothesis (\ref{item:hyp3}),
  $\hat{v}\in\dco{\SS}{\wbs}$. Hence $v^A_k\in\dco{\SS}{\wbs}$ too.
  \smallskip
  
  Next, let us expand that result by showing that all states reachable
  in $G_\tau$ are covered by some state in $\SS\cup\tau(\SS)$,
  i.e. that $\Reach{G_{\tau}} \subseteq \dc{\SS\cup\tau(\SS)}$. To do
  so, it is sufficient to show that each $v_B\in\Reach{G_\tau}\cap
  V_B\in \dc{\tau(\SS)}$. Since $v_B\in\Reach{G_\tau}$, there is
  $v_A\in\Reach{G_\tau}\cap V_A$ s.t. $(v_A,v_B)\in E$. Hence, $v_A\in
  \dc{\SS}$ by the arguments above. Thus, since $\tau$ is a
  $\wbs$-concretisation of $\sigma|_\SS$, there is $\overline{v}\in
  \SS$ s.t. $\overline{v}\wbs v_A$ and
  $\tau(v_A)\wbs\tau(v_B)$. Hence, $v_B\in \dco{\tau(\SS)}{\wbs}$.

  We conclude the proof by observing that that $\tau(\SS) =
  \sigma(\SS)$, since $\tau$ is a $\wbs$-concretisation of
  $\sigma|_\SS$. Hence, $\dc{\SS\cup\tau(\SS)}=
  \dc{\SS\cup\sigma(\SS)}$. Moreover, hypothesis~(\ref{item:hyp1})
  implies that $\dc{\SS\cup\sigma(\SS)}\cap \bad=\emptyset$, by
  definition of tba-simulation (compatible with $\bad$). Hence, since
  $\Reach{G_{\tau}} \subseteq
  \dc{\SS\cup\tau(\SS)}=\dc{\SS\cup\sigma(\SS)}$ by the arguments
  above, we conclude that $\Reach{G_{\tau}}\cap\bad=\emptyset$, and
  thus, that $\tau$ is winning.\qed
\end{proof}

This proposition allows us to identify families of sets of states on
which $\star$-strategies can be defined. One of the sets that
satisfies the conditions of Proposition~\ref{th-monog} is the maximal
antichain of reachable $A$-states, for a given winning strategy
$\sigma$:

\begin{theorem}\label{the:can-play-on-max-antichain}
  Let $G=(V_A,V_B,E,I,\bad)$ be a finite turn-based safety game,
  $\wbs$ be a tba-simulation relation for $G$.  Let $\sigma$ be a
  winning strategy and $\WR{\sigma}$ be a maximal $\wbs$-antichain on
  $\mathsf{Reach}(G_{\sigma})\cap V_A$, then the $\star$-strategy
  $\sigma|_{\WR{\sigma}}$ is $\wbs$-winning.
\end{theorem}
\begin{proof}
  It is straightforward to verify that $\sigma$ and $\WR{\sigma}$ satisfy the
  three properties of Proposition~\ref{th-monog}. Indeed:
  \begin{inparaenum}[(i)]
  \item $\WR{\sigma} \cup \sigma(\WR{\sigma}) \subseteq
    \mathsf{Reach}(G_{\sigma})$ by definition of $\WR{\sigma}$ and
    $\Reach{G_{\sigma}}\cap\bad =\emptyset$ because $\sigma$ is
    winning. Hence $(\WR{\sigma} \cup \sigma(\WR{\sigma}))
    \cap\bad=\emptyset$.
  \item $I \in \Reach{G_{\sigma}}\cap V_A$ and $\Reach{G_{\sigma}}\cap
    V_A\subseteq \dc{\WR{\sigma}}$ by definition, hence $I\in
    \dc{\WR{\sigma}}$.
  \item $\Succ\sigma(\WR{\sigma}))\subseteq \Reach{G_{\sigma}}$ and
    $\Reach{G_{\sigma}}\subseteq \dc{\WR{\sigma}}$ by definition,
    hence $\Succ{\sigma(\WR{\sigma})}\subseteq \dc{\WR{\sigma}}$. \qed
  \end{inparaenum}
\end{proof}

\section{Efficient computation of succinct winning strategies\label{sec:effic-comp-succ}}

\paragraph{The original OTFUR algorithm} 
The OTFUR algorithm~\cite{CDFLL-concur05} is an efficient, on-the-fly
algorithm to compute a winning strategy in a finite
\emph{reachability} game (it is thus easy to adapt it to solve safety
games). We sketch the main ideas behind this algorithm, and refer the
reader to~\cite{CDFLL-concur05} for a comprehensive
description\footnote{See Appendix~\ref{sec:orig-otfur-algor} for the
  original algorithm.}. The intuition of the approach is to combine a
forward exploration from the initial state with a backward propagation
of the information when a losing state is found.  During the forward
exploration, newly discovered states are assumed winning until they
are declared losing for sure. Whenever a losing state is identified
(either because it is $\bad$, or because $\bad$ is unavoidable from
it), the information is back propagated to predecessors whose status
could be affected by this information. A bookkeeping function ${\sf
  Depend}$ is used for that purpose: it associates, to each state $v$,
a list ${\sf Depend}(v)$ of edges that need to be re-evaluated should
$v$ be declared losing.  The main interest of this algorithm is that
it works \emph{on-the-fly} (thus, the arena does not need to be fully
constructed before the analysis), and avoids, if possible, the entire
traversal of the arena.  In this section, we propose an optimized
version of OTFUR for games equipped with tba-simulations.
Before this, we prove that, when a safety game is equipped with a
tba-simulation $\wbs$, then its set of winning states is
\emph{$\wbs$-downward closed}. This property will be important for the
correctness of our algorithm.

We start by an auxiliary lemma that relates tba-simulations and
computation of the attractor set:

\begin{lemma}\label{lem-li}
  Let $G$ be a finite safety game, and let $\wbs$ be an tba-simulation
  for $G$. Then, for all $(v_1,v_2)\in\wbs$ and for all $i \in
  \mathbb{N}$: $v_2 \in \Attr_i$ implies $v_1 \in \Attr_i$.
\end{lemma}
\begin{proof}
  The proof is by induction on $i$.
  
  \noindent{\textbf{Base case: $i=0$}} By definition,
  $\Attr_0=\bad$. However $v_2\in\mathsf{Bad}$ and $v_1\wbs v_2$ imply
  $v_1\in\mathsf{Bad}=\Attr_0$, by definition of tba-simulations.

  \noindent{\textbf{Inductive case: $i>0$}} Assume $v_2 \in
  \Attr_i$. There are two cases to study.
  \begin{itemize}
  \item In the case where $v_2 \in V_B$: first of all $v_1\wbs v_2$
    implies $v_1 \in V_B$. Moreover by definition of $\Attr_i$, there
    exists $v'_2 \in V_A$ such that $(v_2,v'_2)\in E$ and $v'_2 \in
    \Attr_{i-1}$. Then by definition of the tba-simulation relation,
    there exists $v'_1 \in V_A$ such that $(v_1,v'_1)\in E$ and
    $v'_1\wbs v'_2$. By induction hypothesis, $v'_1\in \Attr_{i-1}$
    and hence by definition of $\Attr_i$, $v_1 \in \Attr_i$.
  \item In the case where $v_2 \in V_A$: first of all $v_1\wbs v_2$
    implies $v_1 \in V_A$. Moreover by definition of $\Attr_i$, for
    all $v'_2 \in V_B$ such that $(v_2,v'_2)\in E$, $v'_2 \in
    \Attr_{i-1}$.  On the other hand by definition of the
    tba-simulation relation, for all $v'_1\in V_B$ such that
    $(v_1,v'_1)\in E$, then there exists $v''_2 \in V_B$ such that
    $(v_2,v''_2)\in E$ and $v'_1\wbs v''_2$. As a consequence $v''_2
    \in \Attr_{i-1}$, and thus by induction hypothesis $v'_1\in
    \Attr_{i-1}$. Hence by definition of $\Attr_i$, $v_1 \in
    \Attr_i$.\qed
  \end{itemize}
\end{proof}

Then:

\begin{proposition}\label{th-down}
  Let $G$ be a finite turn-based safety game, and let $\wbs$ be a
  tba-simulation for $G$. Then the set $\win$ of winning states in $G$
  is downward closed for $\wbs$.
\end{proposition}
\begin{proof}[Proposition~\ref{th-down}]
  As a consequence of Lemma~\ref{lem-li}, for all $(v_1,v_2)\in\wbs$,
  $v_2 \in \Attr_\bad$ implies $v_1 \in \Attr_\bad$.  Since $\win$ is
  the complementary of $\Attr_\bad$, $v_1 \in \win$ implies $v_2 \in
  \win$, that is $\win$ is downwards closed for $\wbs$.\qed
\end{proof}

\begin{algorithm}\label{algonous}
    \caption{The {\sf OTFUR} optimized for
    games with a tba-simulation}\label{algo:otfur_plusplus} 
  \KwData{A safety game $G=(V_A,V_B,E,I,\bad)$}
  ${\sf Passed} := \{I\}$ ; ${\sf Depend}(I):= \varnothing$ \; 
  ${\sf AntiMaybe} := \{I\}$ ; ${\sf AntiLosing} := \{\}$ \;
  ${\sf Waiting} := \{(I,v')\mid v'\in\minac{\Succ{I}}\}$ \;

   \While{${\sf Waiting} \neq \varnothing \wedge I\notin \uparrow{\sf AntiLosing}$}{
    $e=(v, v') := {\it pop}({\sf Waiting})$ \;
    \If{$v\notin \uparrow{\sf AntiLosing}$}{
      \uIf{$v\in \downarrow{\sf AntiMaybe}\setminus {\sf AntiMaybe}$}{
        ${\sf choose}\; v_m \in {\sf AntiMaybe}$ s.t. $v_m \wbs v$ \;
        ${\sf Depend}[v_m]:={\sf Depend}[v_m] \cup \{e\}$ \;
      }
      \Else{
        \uIf{$v'\in \downarrow{\sf AntiMaybe}$}{
          \If{$v'\notin {\sf AntiMaybe}$}{
            ${\sf choose}\; v_m \in {\sf AntiMaybe}$ s.t. $v_m \wbs v'$ \;
            ${\sf Depend}[v_m]:={\sf Depend}[v_m] \cup \{e\}$ \;
          }
        }
        \Else{
          \uIf{$v'\not\in {\sf Passed}$}{
            ${\sf Passed}:= {\sf Passed}\cup \{v'\}$ \;
            \uIf{$v' \notin  \uparrow{\sf AntiLosing}$}{
              \uIf{$ (v'\in \mathsf{Bad})$} {
                ${\sf AntiLosing}:= \lfloor{\sf AntiLosing} \cup \{v'\}\rfloor$ \;
                ${\sf Waiting}:= {\sf Waiting} \cup \{e\}$ ; \tcp{reevaluation of $e$}
              }
              \Else{
                ${\sf Depend}[v']:=\{(v,v')\}$ \;
                ${\sf AntiMaybe}:=\lceil{\sf AntiMaybe}\cup \{v'\}\rceil$ \;
                \uIf{$v\in V_A$}{
                  ${\sf Waiting}:= {\sf Waiting} \cup \{(v',v'')\mid v'\in\minac{\Succ{v'}}\}$ \;
                }
                \Else{
                  ${\sf Waiting}:= {\sf Waiting} \cup \{(v',v'')\mid v'\in\maxac{\Succ{v'}}\}$ \;
                }
              }
            }
            \Else(\tcp*[h]{reevaluation of $e$}){
              ${\sf Waiting}:= {\sf Waiting} \cup \{e\}$ ; 
            }
          }
          \Else(\tcp*[h]{reevaluation}){
            ${\sf Losing}^{*} := \begin{array}[t]{ll} &v\in V_A
              \wedge \bigwedge_{v''\in \min({\sf Succ}(v))}(v''\in \uparrow{\sf AntiLosing})\\
              \vee & v\in V_B\wedge \bigvee_{v''\in \max({\sf Succ}(v))}(v''\in \uparrow{\sf AntiLosing})
            \end{array}$ \;
            \uIf{${\sf Losing}^*$}{
              ${\sf AntiLosing}:= \lfloor{\sf AntiLosing} \cup \{v\}\rfloor$ \;
              ${\sf AntiMaybe}:= \lceil{\sf Passed} \setminus \uc{\sf AntiLosing}\rceil$ \;
              \tcp{back propagation}
              ${\sf Waiting}:= {\sf Waiting} \cup {\sf
                Depend}[v]$ ; 
            }
            \Else{
              \lIf{$\neg {\sf Losing}[v']$}{
                ${\sf Depend}[v']:= {\sf Depend}[v']\cup\{e\}$ 
              }
            }
          }
        }
      }
    }
  }
  
  \KwRet{$I\notin \uparrow{\sf AntiLosing}$}
\end{algorithm}

\paragraph{Optimised OTFUR} Let us discuss Algorithm~\ref{algonous},
an optimised version of OTFUR for the construction of $\wbs$-winning
$\star$-strategies in games with tba-simulations. Its high-level
principle is the same than in the original OTFUR, i.e. forward
exploration and backward propagation. At all times, it maintains
several sets:
\begin{inparaenum}[(i)]
\item {\sf Waiting} that stores edges waiting to be explored;
\item {\sf Passed} that stores nodes that have already been explored;
  and
\item {\sf AntiLosing} and {\sf AntiMaybe} which represent, by means
  of antichains (see discussion below) a set of surely losing states
  and a set of possibly winning states respectively.
\end{inparaenum}
The main {\bf while} loop runs until either no more edges are waiting,
or the initial state $I$ is surely losing. An iteration of the loop
first picks an edge $e=(v,v')$ from {\sf Waiting}, and checks whether
the exploration of this edge can be postponed (line $7$--$15$, see
discussion of the optimisations hereunder). Then, if $v'$ has not been
explored before (line $16$), cannot be declared surely losing (line
$18$), and does not belong to $\bad$ (line $19$), it is explored
(lines $23$--$28$). When $v'$ is found to be losing, $e$ is put back
in ${\sf Waiting}$ for back propagation (lines $21$ or $30$). The
actual back-propagation is performed at lines $32$--$38$ and triggered
by an edge $(v,v')$ s.t. $v'\in {\sf Passed}$.  Let us highlight the
three optimisations that prune the state space using a tba-simulation
$\wbs$ on the game states:
\begin{enumerate}
\item By the properties of $\wbs$, we can explore only the
  $\wbs$-minimal (respectively $\wbs$-maximal) successors of each $A$
  ($B$) state (see lines $3$, $26$ and $28$). We also consider maximal
  and minimal elements only when evaluating the status of a node in
  line $32$.
\item By Proposition~\ref{th-down}, the set of winning states in the
  game is downward-closed, hence the set of losing states is
  upward-closed, and we store the set of states that are losing for
  sure as an antichain ${\sf AntiLosing}$ of minimal losing states.
\item Symmetrically, the set of \emph{possibly winning states} is
  stored as an antichain ${\sf AntiMaybe}$ of maximal states. This set
  allows to postpone, and potentially avoid, the exploration of some
  states: assume some edge $(v,v')$ has been popped from {\sf
    Waiting}. Before exploring it, we first check whether either $v$
  or $v'$ belongs to $\dc{\sf AntiMaybe}$ (see lines $7$ and $11$). If
  yes, there is $v_m\in{\sf AntiMaybe}$ s.t. $v_m\wbs v$
  (resp. $v_m\wbs v'$), and the exploration of $v$ ($v'$) can be
  postponed. We store the edge $(v,v')$ that we were about to explore
  in ${\sf Depend}[v_m]$, so that, if $v_m$ is eventually declared
  losing (see line $36$), $(v,v')$ will be re-scheduled for
  exploration.  Thus, the algorithm can stop when all maximal $A$
  states have a successor that is covered by a non-losing one.
\end{enumerate}
Observe that optimisations 1 and 2 rely only the upward closure of the
losing states and were present in the antichain algorithm
of~\cite{FJR-fmsd11}. Optimisation 3 is original and exploits more
aggressively the notion of tba-simulation. It allows to keep at all
times an antichain of potentially winning states, which is crucial to
compute efficiently a winning $\star$-strategy. If, at the end of the
execution, $I\not\in\uc{\sf AntiLosing}$, we can extract from ${\sf
  AntiMaybe}$ a winning $\star$-strategy $\hsg$ as follows. For all
$v\in {\sf AntiMaybe}\cap V_A$, we let $\hsg(v)= v'$ such that $v'
\in\Succ{v}\cap\dc{\sf AntiMaybe}$. For all $v\in V_A \setminus {\sf
  AntiMaybe}$, we let $\hsg(v)=\star$. Symmetrically, if $I\in \uc{\sf
  AntiLosing}$, there is no winning strategy for $A$.

\paragraph{Correctness}
The correctness of Algorithm~\ref{algonous} is given by the following
theorem.
\begin{theorem}\label{th:algo-correct}
  When called on game $G$, Algorithm~\ref{algonous} always
  terminates. Upon termination, either $I\in \uc{\sf AntiLosing}$ and
  there is no winning strategy for $A$ in $G$, or $\hsg$ is a
  $\wbs$-winning $\star$-strategy.
\end{theorem}

We split the proof of Theorem~\ref{th:algo-correct} into two main
propositions establishing respectively termination and soundness of
the algorithm. The proof of soundness relies on auxiliary lemmata
establishing invariants of the algorithm. In those proofs, we rely on
the following notations.  First of all, let us denote by
$\textit{Name}_i$ the state of the set $\textit{Name}$ at the $i$-th
iteration of the ${\bf while}$ in Algorithm~\ref{algonous}.  Let us
define two notations: ${\sf StateWaiting}_i = \{v| \exists v',
(v,v')\in {\sf Waiting}_i \mbox{ or } (v',v)\in {\sf Waiting}_i\}$,
${\sf Visited}_i = \{v | \exists j\le i, v \in {\sf
  StateWaiting}_j\}$, and $S_i = {\sf AntiMaybe}_{i} \cup {\sf
  StateWaiting}_i$.  In words, ${\sf StateWaiting}_i$ is the set of
states which appear in ${\sf Waiting}$ at the $i$-th iteration of {\bf
  while}, ${\sf Visited}_i$ is the set of the states which have
appeared in ${\sf Waiting}$ at some iteration before the $i$-th one,
and $S_i$ is the set of the states which appear in ${\sf Waiting}_i$
or which belong to ${\sf AntiMaybe}_{i}$. Finally, we denote by ${\sf
  AntiMaybe}$, ${\sf Visited}$, ${\sf StateWaiting}$, ${\sf
  AntiLosing}$ and ${\sf Waiting}$ the state of those sets at the end
of the execution of the algorithm.

\begin{proposition}[Termination]
  Algorithm~\ref{algonous} always terminates
\end{proposition}
\begin{proof}
  In order to prove the termination of Algorithm~\ref{algonous}, we
  simply need to prove that the ${\bf while}$ loop cannot be iterated
  infinitely because each iteration of this loop takes finitely many
  steps.  Let us prove it by contradiction.  Let us assume that there
  is an infinite execution of Algorithm~\ref{algonous}. Observe that
  for all indices $i<j$, ${\sf Passed}_{i} \subseteq {\sf Passed}_j$
  and $\uc{\sf AntiLosing}_{i} \subseteq \uc{\sf
    AntiLosing}_j$. Hence, let $K$ be s.t. for all $i\ge K$, ${\sf
    Passed}_{i} = {\sf Passed}_K$ and $\uc{\sf AntiLosing}_{i} =
  \uc{\sf AntiLosing}_K$, i.e. after $K$ steps, ${\sf Passed}$ and
  ${\sf AntiLosing}$ stabilise and remain the same along the rest of
  the infinite run. Such a $K$ necessarily exists because there are
  finitely many states in the arena. Now, we can easily check in the
  code that if ${\sf Passed}$ and $\uc{\sf AntiLosing}$ are not
  modified in an iteration $i$ (i.e., ${\sf Passed}_{i}={\sf
    Passe}_{i-1}$ and $\uc{\sf AntiLosing_i}=\uc{\sf
    AntiLosing_{i-1}}$) then ${\sf Waiting}$ decreases strictly,
  i.e. ${\sf Waiting}_{i} \subseteq {\sf Waiting}_{i-1}$. Since at
  step $K$, ${\sf Waiting}_K$ is necessarily finite, and since ${\sf
    AntiLosing}$ and ${\sf Passed}$ stay constant from step $K$, there
  is a step $K'\geq K$ s.t. ${\sf Waiting}_{K'}=\emptyset$. However,
  this implies that the algorithms stops at step
  $K'$. Contradiction.\qed
\end{proof}

Let us now turn our attention to soundness. We first establish several
invariants of the algorithm:

\begin{lemma}\label{lm-inv}
  Algorithm~\ref{algonous} admits the three following
  loop-invariants: \begin{align*}
    {\sf Inv}^1_i &: {\sf Visited}_i\setminus{\sf StateWaiting}_i  \subseteq \dc{\sf AntiMaybe}_i \cup \uc{\sf AntiLosing}_i\\
    {\sf Inv}^2_i &: \forall v \in \dc{\sf AntiMaybe}_i \cap V_A,
    \exists v'\in {\sf Succ}(v):\\
    &\qquad v' \in \dc{\sf AntiMaybe}_i \lor (v,v')\in {\sf
      Depend}_i[\dc{\sf AntiMaybe}_i]
    \cup {\sf Waiting}_i\\
    {\sf Inv}^3_i &: \forall v \in \dc{\sf AntiMaybe}_i
    \cap V_B, \forall v'\in {\sf Succ}(v):\\
    &\qquad v' \in \dc{\sf AntiMaybe}_i \lor (v,v')\in {\sf
      Depend}_i[\dc{\sf AntiMaybe}_i] \cup {\sf Waiting}_i
  \end{align*}
\end{lemma}
\begin{proof}
  We start by observing that the sets $\uc{\sf AntiLosing}$ and ${\sf
    Visited}$ increase monotonically along the execution of the
  algorithm. Then let us prove each invariant separatly
  \begin{itemize}
  \item[(${\sf Inv}^1_i$)] At $i=0$, all the states in ${\sf Visited}$
    are in ${\sf StateWaiting}$, hence the initialization is trivial.

    Let us now assume that ${\sf Inv}^1_i $ is true for a given $i$
    and let us prove that ${\sf Inv}^1_{i+1}$ is also true.  Observe
    that the definition of ${\sf Visited}_i$ depends only on ${\sf
      StateWaiting}$: ${\sf Visited}_i=\cup_{j\leq i} {\sf
      StateWaiting}_j$. This implies in particular that ${\sf
      StateWaiting}_i\subseteq {\sf Visited}_i$ for all $i$. Thus, we
    only need to consider the modification of ${\sf StateWaiting}$
    during one iteration to prove this invariant. At each iteration of
    the loop, one edge is taken from ${\sf Waiting}$, and several
    edges are potentially added. Thus, the sets ${\sf
      StateWaiting}_{i+1}\setminus {\sf StateWaiting}_i$ and ${\sf
      StateWaiting}_{i}\setminus {\sf StateWaiting}_{i+1}$ are both
    potentially non-empty. This discussion allows to conclude that, for
    all states $v\in{\sf Visited}_{i+1}\setminus{\sf
      StateWaiting}_{i+1}$, we only need to consider two cases:
    \begin{inparaenum}[(i)]
    \item either $v\in {\sf Visited}_i\setminus {\sf StateWaiting}_i$;
    \item or $v\in {\sf StateWaiting}_{i}\setminus {\sf
        StateWaiting}_{i+1}$
    \end{inparaenum}
    since by definition, nodes in ${\sf StateWaiting}_{i+1}\setminus
    {\sf StateWaiting}_i$ are not in ${\sf Visited}_{i+1}\setminus{\sf
      StateWaiting}_{i+1}$.

    \begin{enumerate}
    \item Let $v$ be in $\big({\sf Visited}_{i+1}\setminus{\sf
        StateWaiting}_{i+1}\big)\cap \big({\sf Visited}_i\setminus
      {\sf StateWaiting}_i\big)$. Since $v\in {\sf Visited}_i\setminus
      {\sf StateWaiting}_i$, $v$ is in $\dc{\sf AntiMaybe}_{i}\cup
      \uc{\sf AntiLosing}_{i}$, by induction hypothesis.
      \begin{enumerate}
      \item If $v\in \uc{\sf AntiLosing}_{i}$, then $v\in \uc{\sf
          AntiLosing}_{i+1}$ (see the begining of the proof), hence,
        $v\in \dc{\sf AntiMaybe}_{i+1}\cup \uc{\sf AntiLosing}_{i+1}$,
        and $v$ respects the invariant.
      \item Otherwise, $v\in\dc{\sf AntiMaybe}_{i}$. Then, either
        $v\in\dc{\sf AntiMaybe}_{i+1}$ (i.e., $\dc{\sf AntiMaybe}$ has
        not decreased), which implies $v\in \dc{\sf
          AntiMaybe}_{i+1}\cup \uc{\sf AntiLosing}_{i+1}$, and $v$
        respects the invariant. Or $v\not\in\dc{\sf AntiMaybe}_{i+1}$,
        i.e. {\sf AntiMaybe} has decreased. This can occur only in
        line 35, but, in this case, all states that are removed from
        $\dc{\sf AntiMaybe}$ are inserted in ${\sf
          StateWaiting}_{i+1}$ (actually, edges containing those
        states are inserted in ${\sf Waiting}_{i+1}$ in line 36),
        which contradicts our hypothesis that $v\in{\sf
        Visited}_{i+1}\setminus{\sf StateWaiting}_{i+1}$.
      \end{enumerate}
    \item Otherwise, let $v$ be in $v\in {\sf
        StateWaiting}_{i}\setminus {\sf StateWaiting}_{i+1}$.  This
      can occur only because an edge $(v_1,v_2)$ has been popped in
      line 5, with either $v=v_1$ or $v=v_2$. We consider those two
      cases separately.
      \begin{enumerate}\item 
        If $v=v_1$, then, $v$ is necessarily in ${\sf Passed}$ because
        lines $26$ and $28$ are the only lines where new edges are
        built and, if we execute one of these lines, adding an edge of
        the form $(v_1,v_2)$ implies that $v_1$ has been added in
        ${\sf Passed}$ on line $17$.  One can also check that, in this
        conditional, $v_1$ is either added to ${\sf AntiMaybe}$ or to
        ${\sf AntiLosing}$. Thus, when a state is in ${\sf Passed}$,
        then it will always be in $\dc{\sf AntiMaybe}\cup \uc{\sf
          AntiLosing}$. Hence $v$ belongs to $\dc{\sf
          AntiMaybe}_{i+1}\cup \uc{\sf AntiLosing}_{i+1}$.
      \item Otherwise, if $v=v_2$, then either $v \in \dc{\sf
          AntiMaybe}_{i+1}\cup \uc{\sf AntiLosing}_{i+1}$, or it is
        not already passed and the conditional on line $16$ is
        satisfied.  As a consequence, at the end of the iteration, $v
        \in {\sf AntiMaybe}_{i+1}$ or $v\in {\sf AntiLosing}_{i+1}$.
      \end{enumerate}
    \end{enumerate}

  \item[(${\sf Inv}^2_i$)] At $i=0$, the only state in ${\sf
      AntiMaybe}$ is $I$ and all the edges of the form $(I,v)$ are in
    ${\sf Waiting}$, hence the initialization is trivial.

    Let us assume that ${\sf Inv}^2_i $ is true for a given $i$ and
    let us prove that ${\sf Inv}^2_{i+1}$ is also true.  To do so, we
    have to inspect all the cases which could make ${\sf Inv}^2_{i+1}$
    false when ${\sf Inv}^2_i $ is true. For the sake of clarity, we
    use the \checkmark symbol to mean that a case is closed.

    \begin{enumerate}
    \item Let us assume that there exist $v \in \dc{\sf AntiMaybe}_i$
      and $v'\in {\sf Succ}(v)$ such that $v' \in \dc{\sf
        AntiMaybe}_i\setminus \dc{\sf AntiMaybe}_{i+1}$. The strict
      inclusion of $\dc{\sf AntiMaybe}_{i+1}$ in $\dc{\sf
        AntiMaybe}_{i}$ implies that line $35$ has been executed
      during the $i$-th iteration of the loop. Then, either $v'$ is
      the state which is put in ${\sf AntiLosing}_{i+1}$ on line $34$
      and $(v,v') \in {\sf Depend}_i[v']$ (see line 36) and thus
      $(v,v') \in {\sf Waiting}_{i+1}$ (\checkmark), or $v'\in {\sf
        Passed}$ and by construction on line $35$, $v'$ necessarily
      belongs to $\dc{\sf AntiMaybe}_{i+1}$ which contradicts our
      hypothesis (\checkmark), or $v'\notin {\sf Passed}$ and there
      are two further possible cases. Indeed, if $v'\notin {\sf
        Passed}$, either $(v,v') \in {\sf Waiting}_i$ (\checkmark), or
      the exploration of edge $(v,v')$ has been postponed (lines
      7--14) because there was a $\wbs$-greater state than $v'$ in
      ${\sf AntiMaybe}$.  In this latter case, either there is again a
      $\wbs$-greater state than $v'$ in ${\sf AntiMaybe}$
      (\checkmark), or $(v,v')$ necessarily belongs to ${\sf
        Depend}_i[w]$ where $w$ is the state put in ${\sf
        AntiLosing}_{i+1}$ (\checkmark).
    \item Let us assume that there exist $v \in \dc{\sf AntiMaybe}_i$
      and $v'\in {\sf Succ}(v)$ such that $(v,v') \in {\sf
        Waiting}_i\setminus {\sf Waiting}_{i+1}$.  This implies that
      $(v,v')$ has been popped (line 5) at the begining of the $i+1$th
      iteration, and there are three possible cases.
      \begin{enumerate}
      \item $(v,v')\in {Depend}_{i+1}[\dc{\sf AntiMaybe}_{i+1}]$ (line
        $9$ or $14$) (\checkmark)
      \item $v'\in {\sf Passed}_{i+1}\setminus {\sf Passed}_{i}$ (line
        17). Then, either (lines 20--21) $v' \in \uc{\sf
          AntiLosing}_{i+1}$ and $(v,v')\in {\sf Waiting}_{i+1}$ which
        contradicts our hypothesis (\checkmark), or (line 24) $v'\in
        \dc{\sf AntiMaybe}_{i+1}$ (\checkmark).
      \item \label{item:1} $(v,v')$ goes in the "reevaluation"
        part. Then, either there is no successor of $v$ outside of
        $\uc{\sf AntiLosing}_{i}$ hence $v \notin {\sf
          AntiMaybe}_{i+1}$ (\checkmark), or there exists a successor
        $w$ of $v$ which is not in $\uc{\sf AntiLosing}_{i}$. In this
        latter case, either $(v,w)$ has not already been treated and
        $(v,v')\in {\sf Waiting}_{i+1}$ (\checkmark), or $w \in
        \dc{\sf AntiMaybe}_{i+1}$ (\checkmark).
      \end{enumerate}
    \item Let us assume that there exist $v \in \dc{\sf AntiMaybe}_i$
      and $v'\in {\sf Succ}(v)$ such that $(v,v') \in {\sf
        Depend}_{i}[\dc{\sf AntiMaybe}_{i}]\setminus {\sf
        Depend}_{i+1}[\dc{\sf AntiMaybe}_{i+1}]$.  This implies that a
      state $w$ of ${\sf AntiMaybe}_{i}$ is found losing and added to
      ${\sf AntiLosing}_{i+1}$ (line $34$). One can assume that $w
      \neq v'$ because the case $w=v'$ is treated in case~1
      above. Then, since $(v,v') \notin {\sf Depend}_{i+1}[\dc{\sf
        AntiMaybe}_{i+1}]$, $(v,v') \in {\sf Depend}_{i}[w]$ because,
      $w$ is necessarily the only state in ${\sf Passed}_i$ which is
      in $\dc{\sf AntiMaybe}_{i}\setminus \dc{\sf AntiMaybe}_{i+1}$
      (this holds since ${\sf AntiMaybe}_{i+1}=\break \minac{{\sf
          Passed}_{i+1}\setminus{\sf Losing}_{i+1}}$, and $w$ is the
      only state identified as losing during iteration $i+1$), and to
      have ${\sf Depend}_{i}[w']\neq \emptyset$, it is necessary that
      $w'$ is in ${\sf Passed}$ (\checkmark).
    \item Let us assume that there exist $v \in \dc{\sf
        AntiMaybe}_{i+1}\setminus \dc{\sf AntiMaybe}_{i}$.  Then, line
      $26$ or line $28$ is executed and for all $v'\in
      \maxac{\Succ{v}}$: $(v,v') \in {\sf Waiting}_{i+1}$.
    \end{enumerate}

  \item[(${\sf Inv}^2_i$)] The proof of $(3)$ can be done with the
    disjunction of cases of the proof of $(2)$. The only difference is
    for the case~\ref{item:1}.  Indeed, the case is simpler because
    when $(v,v')$ goes in the "reevaluation" part, it is not necessary
    to consider other successors, because $v$ must have only
    non-losing successors. Then, either there is a successor of $v$ in
    $\uc{\sf AntiLosing}_{i}$ hence $v \notin \dc{\sf
      AntiMaybe}_{i+1}$ (\checkmark), or $v'$ is in ${\sf
      Passed}_{i+1}\setminus \uc{\sf AntiLosing}_{i}$ and thus belongs
    to $\dc{\sf AntiMaybe}_{i+1}$ (\checkmark). \qed
  \end{itemize}
\end{proof}

\begin{lemma}\label{lm-inv4}
Algorithm~\ref{algonous} admits the following loop-invariant:
\begin{align*} {\sf Inv}^4_i &: {\sf AntiLosing}_i \subseteq {\sf
    Losing}
\end{align*}
\end{lemma}
\begin{proof}
  At $i=0$, the ${\sf AntiLosing}$ is $\empty$, hence the
  initialization is trivial.

  Let us assume that ${\sf Inv}^4_i $ is true for a given $i$ and let
  us prove that ${\sf Inv}^4_{i+1}$ is also true.  We simply have to
  check that states in $\uc{\sf AntiLosing}$ at the $i$-th iteration
  are losing.

  There are two lines where states are added to $\uc{\sf AntiLosing}$:
  lines $20$ and $34$.
  \begin{itemize}
  \item (line $20$) A state $v'$ can be added in ${\sf AntiLosing}$ on
    line $20$. In this case, the conditional on line $19$ ensures that
    the state is in ${\sf Bad}$, hence $v'$ is in ${\sf Losing}$.
  \item (line $34$) A state $v$ can be added in ${\sf AntiLosing}$ on
    line $34$. In this case, the definition of $Losing^*$ and the
    conditional on line $33$ ensures that $v$ is in ${\sf
      Losing}$. Indeed, if $v\in V_A$ then all its minimal successors
    are in $\uc{\sf AntiLosing}_i$, hence all its successors are
    losing by induction assumption and thus $v$ is in ${\sf
      Losing}$. Otherwise, $v\in V_B$ and it has a successor in
    $\uc{\sf AntiLosing}_i$, hence one of its successors is losing by
    induction assumption and thus $v$ is in ${\sf Losing}$. \qed
  \end{itemize}
\end{proof}

\begin{lemma}
  Algorithm~\ref{algonous} admits the following loop-invariant:
\begin{align*} {\sf Inv}^5_i &: \forall \overline{v}, \forall (v,v')\in {\sf Depend}_i[\overline{v}]: \overline{v}\wbs v'\textrm{ or } (\overline{v}\wbs v \textrm{ and } \overline{v}\neq v)
\end{align*}
\end{lemma}
\begin{proof}
  The invariant is easily established by checking lines 9, 14 and 23,
  which are the only lines where an edge is added to {\sf Depend}.\qed
\end{proof}

We are now ready to prove soundness of the algorithm:
\begin{proposition}[Soundness]
  When Algorithm~\ref{algonous} terminates, either $I\in \uc{\sf
    AntiLosing}$ and there is no winning strategy for $A$ in this
  game, or $\sigma$ is a $\wbs$-winning $\star$-strategy.
\end{proposition}
\begin{proof}
  We first consider the case where Algorithm~\ref{algonous} ends with
  $I \notin \uc{\sf AntiLosing}$.  In particular, at the end of the
  execution, ${\sf Waiting}$ is empty.

  Let us first show that $\hsg$ is well-defined, i.e. for all
  $v\in {\sf AntiMaybe}\cap V_A$, there is $v'\in\Succ{v}$
  s.t. $v'\in\dc{\sf AntiMaybe}$. This stems from ${\sf Inv}^1_i$,
  ${\sf Inv}^2_i$ and ${\sf Inv}^5_i$. Indeed, at the end of the
  execution, ${\sf Waiting}$ is empty, hence ${\sf StateWaiting}$ is
  empty. Thus, ${\sf Inv}^2_i$ entails that all $v\in {\sf
    AntiMaybe}\cap V_A$ have a successor $v'$ s.t. either $v'\in
  \dc{\sf AntiMaybe}$ or $(v,v')\in {\sf Depend}_i[\dc{\sf
    AntiMaybe}]$. In the former case, the property is established. In
  the latter case ($(v,v')\in {\sf Depend}_i[\dc{\sf AntiMaybe}]$),
  let $\overline{v}$ be a node from $\dc{\sf AntiMaybe}$
  s.t. $(v,v')\in{\sf Depend}_i(\overline{v})$. By ${\sf Inv}^5_i$,
  either $\overline{v}\wbs v'$, or $\overline{v}\wbs v$ and
  $\overline{v}\neq v$. We first observe that the case
  `$\overline{v}\wbs v$ and $\overline{v}\neq v$' is not possible
  because $v\in {\sf AntiMaybe}$ by hypothesis, and
  $\overline{v}\in\dc{\sf AntiMaybe}$. Thus, $\overline{v}\wbs v'$,
  which implies that $v'\in \dc{\sf AntiMaybe}$ since
  $\overline{v}\in\dc{\sf AntiMaybe}$. Thus, $\hsg$ is well-defined.

  We conclude the proof by invoking Proposition~\ref{th-monog}. To be
  able to apply this proposition, we need a strategy $\sigma$. We let
  $\sigma$ be any concretisation of $\hsg$, and $\SS={\sf
    AntiMaybe}\cap V_A$. The choice of the concretisation of $\hsg$
  does not matter, because the hypothesis required by
  Proposition~\ref{th-monog} are properties of $\sigma(v)$ for states
  $v\in \SS$ only, and $\SS={\sf AntiMaybe}\cap V_A$ is exactly the
  support of $\hsg$. Let us show that $\sigma$ respects the three
  hypothesis of Proposition~\ref{th-monog}, i.e. that:
  \begin{inparaenum}[(i)]
  \item $({\sf AntiMaybe}\cap V_A \cup\hsg({\sf AntiMaybe}\cap
    V_A))\cap \bad = \emptyset$;
  \item $I \in \dco{{\sf AntiMaybe}\cap V_A}{\wbs}$; and
  \item $\mathsf{succ}(\hsg({\sf AntiMaybe}\cap V_A))\subseteq
    \dco{{\sf AntiMaybe}\cap V_A}{\wbs}$.
  \end{inparaenum}
  
  \begin{itemize}
  \item[(i)] By definition of alternating simulation, ${\sf Bad}$ is
    upward closed ($\uc{\sf Bad}={\sf Bad}$). No bad state can be
    added to ${\sf AntiMaybe}$ during the execution, because of the
    conditional in line $19$ which is false when a state is added to
    ${\sf AntiMaybe}$ (line $24$).  Moreover by definition of
    $\hsg$, $\hsg({\sf AntiMaybe}\cap V_A)\subseteq \dc{\sf
      AntiMaybe}$, hence (i) is satisfied.
  \item[(ii)] By assumption, $I \notin \uc{\sf AntiLosing}$. As a
    consequence of ${\sf Inv}^1_i$, $I$ belongs to $\dc{\sf
      AntiMaybe}$, hence (ii) is satisfied.
  \item[(iii)] By definition of $\hsg$, $\hsg({\sf AntiMaybe}\cap
    V_A)\subseteq \dc{\sf AntiMaybe}\cap V_B$. Now, by ${\sf
      Inv}^3_i$, at the end of the execution, ${\sf Succ}(\dc{\sf
      AntiMaybe}\cap V_B)\subseteq \dc{\sf AntiMaybe}\cap
    V_A$. Indeed, at the end of the execution, ${\sf StateWaiting}$ is
    empty and $(v,v') \in {\sf Depend}[\dc{\sf AntiMaybe}]$ implies,
    by ${\sf Inv}^5_i$, that either $v' \in \dc{\sf AntiMaybe}$, or $v
    \in \dc{\sf AntiMaybe}$ which means that there is a state $w$ s.t.
    $w\wbs v$ in ${\sf AntiMaybe}$. In this latter case, by definition
    of the tba-simulation, there exists $\overline{v}\in {\sf
      Succ}(v)$ such that $\hsg(w)\wbs \overline{v}$, therefore
    $\overline{v}\in \dc{\sf AntiMaybe}$.  Hence (iii) is satisfied.
  \end{itemize}
  Hence, by Proposition~\ref{th-monog}, the $\star$-strategy
  $\sigma|_{\SS}$ (with $\SS={\sf AntiMaybe}\cap V_A$) is a
  $\wbs$-winning $\star$-strategy. It is easy to check that
  $\sigma|_{\SS}=\hsg$, by definition of $\sigma$.
  
  \medskip We conclude the proof by considering the case where
  Algorithm~\ref{algonous} ends with $I \in \uc{\sf AntiLosing}$. By
  ${\sf Inv}^4_i$, there is no winning strategy for $A$ in the
  game.\qed
\end{proof}

\paragraph{Why simulations are not sufficient} Let us exhibit two
examples of games equipped with a simulation $\succeq$ which is not a
tba-simulation, to show why tba-simulations are crucial for our
optimisations. In Fig.~\ref{fig:clos} (left), $\bad=\{v_1',v_2'\}$,
and the set of winning states is not $\succeq$-downward closed (gray
states are losing). In the game of Fig.~\ref{fig:clos} (right),
$\bad=\{b_1,b_2\}$ and Algorithm~\ref{algonous} does not develop the
successors of $v'$ (because $v\succeq v'$, and $v\in{\sf AntiMaybe}$
when first reaching $v'$). Instead, it computes a purportedly winning
$\star$-strategy $\hsg$ s.t. $\hsg(v)=v''$ and
$\hsg(v')=\star$. Clearly this $\star$-strategy is not
$\succeq$-winning (actually, there is no winning strategy in this
game).

\begin{figure}[t]
\begin{center}
    \scalebox{0.8}{
      \begin{tikzpicture}[->,>=stealth',shorten >=1pt,auto,node distance=2cm,
                    semithick]

  \tikzstyle{every state}=[text=black]

  \node[state, fill=white] (A) {$v_1$};
  \node[state, fill=white, shape=rectangle] (B) [right of=A, node distance=1.5cm]{$v_1''$};% {$v''_1$};
  \node[state, fill=gray!20] (A') [below of=A, node distance=2.5cm] {$v_2$};
 
  \node[state, shape=rectangle] (A'') [left of=A, node distance=1.5cm] {} ;
  \node[state, shape=rectangle, fill=gray!20] (A''') [left of=A', node distance=1.5cm] {} ;

  \node[state, shape=rectangle, fill=gray!60] (D) [right of=A, node distance=1.5cm, yshift=-1cm] {$v_1'$} ;
  \node[state, shape=rectangle, fill=gray!60] (E) [right of=A', node distance=1.5cm] {$v_2'$} ;

  \node[state, fill=white] (I) [left of=A'', node distance=1cm, yshift=-1.25cm] {$v_0$};

  \path (A) edge (B)
        (A) edge (D)
        (A') edge (E)
        (I) edge (A'')
        (I) edge (A''')
        (A'') edge (A)
        (A''') edge (A')
 %       (D) edge (C)
  %      (E) edge (B')
;
\path[-]
        (A) edge [dashed] node [sloped,yshift=.3cm,xshift=-.3cm] {$\succeq$} (A')
    
        (D) edge [dashed] node [sloped,yshift=.3cm,xshift=-.3cm] {$\succeq$} (E)
;
\end{tikzpicture}
}
\textcolor{white}{blabla}\scalebox{0.8}{
\begin{tikzpicture}[->,>=stealth',shorten >=1pt,auto,node distance=2cm,
                    semithick]

  \tikzstyle{every state}=[text=black]

  \node[state, rectangle, fill=white] (A)  {$v''$};
  \node[state, fill=white] (A0) [left of=A, node distance=2cm] {$v$};
  \node[state, fill=white] (X1) [right of=A, node distance=2cm] {$v'$};
 
  \node[state, rectangle, fill=gray!60] (bad1) [above of=A0, node distance=1.3cm, xshift=0cm] {$b_1$};
  \node[state, rectangle, fill=gray!60] (bad2) [above of=X1, node distance=1.3cm, xshift=0cm] {$b_2$};

 \path (A) edge (X1)
        (-3,0) edge (A0);
     \path    (A0) edge (A)
       (A0)  edge (bad1)
       (X1) edge (bad2);
    \path[-]  (A0) edge [dashed,bend left] node [sloped] {$\succeq$} (X1);
\end{tikzpicture}}
\caption{A simulation and the downward closure are not sufficient to
  apply Algorithm~\ref{algonous}.  }\label{fig:clos}
\end{center}
  \vspace{-.6cm}
\end{figure}
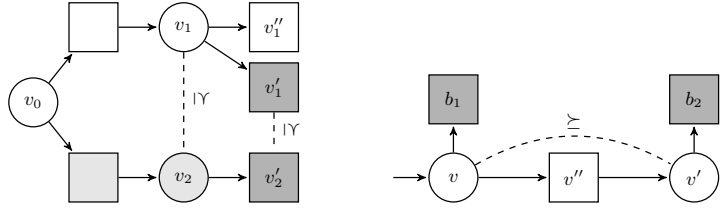

\section{Applications\label{sec:applications}}
To show the relevance of our approach, let us briefly explain how to
apply it to three problems that can be reduced to a safety game: LTL
realisability, real-time scheduler synthesis and determinisation of
timed automata. These three problems share the following
characteristics, that make our technique particularly appealing: 
\begin{inparaenum}[(i)]
\item they have practical applications where an efficient
  implementation of the winning strategy is crucial ;
\item the arena of the safety game is not given explicitly and is at
  least exponential in the size of the problem instance; and
\item they admit natural tba-simulations.
\end{inparaenum}

\paragraph{$A$-deterministic and $\wbs$-monotonic games} To show that
a \emph{tba-simulation} exists in a safety game for which a
\emph{simulation relation} $\wbs$ is already known, we rely on the
notions of \emph{$A$-determinism} and \emph{$\wbs$-monotonicity} of
safety games. Intuitively, this means that player $A$ can always chose
to play \emph{the same set of actions} in each of its states, and that
playing the same action $a$ in two states $v_1\wbs v_2$ yields two
states $v_1'$ and $v_2'$ with $v_1'\wbs v_2'$ \footnote{For example,
  in the urn-filling game (Fig.~\ref{fig:substrgame}), Player $A$ can
  always choose between taking $1$ or $2$ balls, from all states where
  at least $2$ balls are left.}.  Formally, let $G=(V_A,V_B,E,I,\bad)$
be a finite turn-based safety game and $\Sigma$ a finite alphabet.  A
\emph{labeling} of $G$ is a function $\mathsf{lab}: E \rightarrow
\Sigma$. For all states $v\in V_A\cup V_B$, and all $a\in \Sigma$, we
let $\Succlab{v}{a}=\{v'\mid (v,v')\in E\land {\sf lab}(v,v')=a\}$ be
the set of $a$-labeled successors of $v$. Then, $(G,\mathsf{lab})$ is
\emph{$A$-deterministic} iff there is a set of actions
$\Sigma_A\subseteq \Sigma$ s.t. for all $v\in V_A$:
\begin{inparaenum}[(i)]
\item $|\Succlab{v}{a}|=1$ for all $a\in\Sigma_A$ and 
\item $|\Succlab{v}{a}|=0$ for all $a\not\in\Sigma_A$.
\end{inparaenum}
Moreover, a labeling $\mathsf{lab}$ is $\wbs$-\emph{monotonic} (where
$\wbs$ is a simulation relation on the states of $G$) iff for all
$v_1, v_2\in V_A\cup V_B$ such that $v_1\wbs v_2$, for all
$a\in\Sigma$, for all $v_2'\in\Succlab{v_2}{a}$: there is
$v_1'\in\Succlab{v_1}{a}$ s.t. $v'_1\wbs v_2'$. Then:
\begin{theorem}\label{th-red}
  Let $G=(V_A,V_B,E,I,\bad)$ be a finite turn-based safety game, let
  $\wbs$ be a simulation relation on $G$ and let $\mathsf{lab}$ be a
  \emph{$\wbs$-monotonic} labeling of $G$. If $(G,\mathsf{lab})$ is
  \emph{$A$-deterministic}, then $\wbs$ is a tba-simulation relation.
\end{theorem}
\begin{proof}
  Since $\wbs$ is a simulation relation, we only need to prove that
  for all $v_1\in V_A$, for all $v_2$ s.t. $v_1\wbs v_2$, for all
  $v_1'\in\Succ{v_1}$, there is $v_2'\in V_B$ s.t. $v_2'\in\Succ{v_2}$
  and $v_1'\wbs v_2'$. Let $v_1,v_2\in V_A$ be s.t. $v_1\wbs v_2$, and
  let $v_1'$ be a state from $\Succ{v_1}$. Since $G$ is
  $A$-deterministic, there is $v_2'\in \Succ{v_2}$ s.t. ${\sf
    lab}(v_1,v_1')={\sf lab}(v_2,v_2')$. Since $\mathsf{lab}$ is
  $\wbs$-\emph{monotonic}, we also have $v_2\wbs v_2'$. \qed
\end{proof}

Thus, when a safety game $G$ is labeled, $A$-deterministic,
$\wbs$-monotonic and equipped with a simulation relation $\wbs$ that
can be computed directly on the description of the states, our
approach can be applied out-of-the-box.  In this case, the algorithm
of Section~\ref{sec:effic-comp-succ} yields, if it exists, a winning
$\star$-strategy $\hsg$, that can be described by the set of pairs
$(v,a)$ for all $v$ in the support of $\hsg$ (the maximal antichain of
winning reachable states), and where $a$ is the label of the edge
$(v,\hsg(v))$. That is, we can store the \emph{action} one needs to
play from $v$ instead of the \emph{successor} $\hsg(v)$. Observe that
no information needs to be stored for the partial order $\wbs$ that
can be directly computed on the description of the states. Then, a
controller implementing $\hsg$ works as follows: when the system state
is $v$, the controller looks for a pair $(\overline{v},a)$ in the
strategy description with $\overline{v}\wbs v$. It executes action $a$
from the current state $v$, which is possible by $A$-determinism, and
respects the definition of $\wbs$-concretisation by
$\wbs$-monotonicity. Finding $\overline{v}$ should be efficient,
because we expect the antichain of winning reachable states to be
compact. Let briefly explain why this technique applies to the three
cases mentioned above.

\paragraph{LTL realisability} LTL \cite{Pnu77} is a popular logic to
express properties of computer systems. An LTL formula defines a set
of traces, i.e. infinite sequences of valuations of atomic
propositions.  In the LTL realisability problem \cite{PR89}, the set
of atomic propositions is partitioned into \emph{controllable} and
\emph{uncontrollable} ones. The controller and the environment are two
players that compete in a game, where, at each turn, they fix the
valuations of the atomic propositions they own, thereby building a
trace. The play is winning for the controller iff the trace satisfies
a given LTL formula. A formula is \emph{realisable} iff the controller
has a winning strategy in the game.

In~\cite{FJR-fmsd11}, Filiot, Jin and Raskin reduce the realisability
problem to a safety game, whose states are vectors of bounded natural
numbers. They consider the partial order $\succeq$ on the states of
the game, defined as $v\succeq v'$ iff $v[i]\geq v'[i]$ for all
coordinates $i$. They show it is a \emph{simulation relation} and rely
on it to define an efficient antichain algorithm (based on the OTFUR
algorithm) to solve the class of safety games they obtain
from the realisability problem.

As a matter of fact, our technique generalises these
results. Theorem~\ref{th-red} can be invoked to show that $\succeq$ is
a \emph{tba-simulation}. Hence, Algorithm~\ref{algonous} can be use to
solve LTL realisability. As already explained, the antichain algorithm
of~\cite{FJR-fmsd11} contains two of the three optimisations that are
present in Algorithm~\ref{algonous} (see
Section~\ref{sec:effic-comp-succ}). Our results thus provide a general
theory to explain the excellent performance of the technique
of~\cite{FJR-fmsd11}, and have the potential to further improve it.

\paragraph{Multiprocessor real-time scheduler synthesis} We consider
the problem of computing a correct scheduler for a given set of
\emph{sporadic real-time tasks}. A sporadic task $(C,T,D)$ is a
process that repeatedly creates \emph{jobs}, s.t. each job creation
(also called \emph{request}) occurs at least $T$ time units after the
previous one. Each job models a computational payload. It needs at
most $C$ units of CPU time to complete, and must obtain them within a
certain time frame of length $D$ starting from the request (otherwise
the job \emph{misses} its deadline). We assume the tasks run on a
platform of $m$ identical CPUs. A scheduler is a function that
assigns, at all times, jobs to available CPUs. A scheduler is
\emph{correct} iff it ensures that no job ever misses a deadline,
whatever the sequence of requests.

This problem can be reduced to a safety game~\cite{BM-esa10} where the
two players are the scheduler and the coalition of the tasks
respectively.  In this setting, a \emph{winning} strategy for Player
$A$ is a \emph{correct} scheduler. In practice, this approach is
limited by the size of the \emph{game arena}, which is, in general,
exponentially larger than the description of the set of tasks
\cite{BM-esa10}. Nevertheless, we can rely on Theorem~\ref{th-red} to
show that the game admits a tba-simulation relation. Indeed the
simulation relation $\succeq$ introduced in~\cite{GGL-rts13} (to solve
a related real-time scheduling problem using antichain techniques)
naturally induces a simulation relation on the set of states of the
game. An $A$-deterministic and $\succeq$-monotonic labeling is
obtained if we label moves of the environment by the set of tasks
producing a request, and the scheduler moves by a total order on all
the tasks, which is used as a priority function determining which
tasks are scheduled for running.

\paragraph{Determinisation of timed automata}
Timed automata (TA for short) \cite{AD-tcs94} are a well-established
model for real-time systems. TAs extend finite automata with clocks,
that are real-valued variables evolving at the same rate, and that can
be constrained by guards on transitions. Fig.~\ref{fig:ex-game-beg}
presents a timed automaton with one clock $x$, locations $\ell_0$,
$\ell_1$ and $\ell_2$ and the alphabet $\{a,b\}$. A TA naturally
defines a \emph{timed language} and two timed automata are said to be
\emph{equivalent} if they admit the same language.

A TA is deterministic if from every state, at most one transition is
firable for each action. The determinisation of a TA $\mathcal{A}$ is
the construction of a deterministic TA equivalent to $\mathcal{A}$ and
is a crucial operation for several problems such as test generation,
fault diagnosis or more generally all the problems closed to the
complement operation.  Unfortunately, TAs are not determinisable in
general~\cite{AD-tcs94}. Fig.~\ref{fig:ex-game-beg} illustrates the
difficulty: from location $\ell_0$, there are two edges with action
$a$ and guard $0<x<1$, but $x$ is reset on only one of these edges.
Hence, a deterministic version of this TA should have at least two
clocks to keep track of the two possible clock values on these two
branches. Based on this idea, one can build a TA (with loops) for
which the number of clock values that must be tracked simultaneously
cannot be bounded. Furthermore, checking whether a given TA admits a
deterministic version is an undecidable problem~\cite{AD-tcs94}. As a
consequence, only partial algorithms exist for determinisation.

So far, the most general of those techniques has been introduced
in~\cite{BSJK-fossacs11} and consists in turning a TA $\mathcal{A}$
into a safety game $G_{\mathcal{A},(Y,M)}$ (parametrised by a set of
clocks $Y$ and a maximal constant $M$). Then, a deterministic TA
over-approximating $\mathcal{A}$ (with set of clocks $Y$ and maximal
constant $M$), can be extracted from any Player $A$ strategy. If the
strategy is winning, then the approximation is an \emph{exact}
determinisation. The idea of the construction is that Player $A$
actions consists in choosing a good reset policy for the clocks to
avoid states in which the deterministic TA could over-approximate
$\mathcal{A}$.  States of the game can be seen as pairs
$((S,S_\top),r)$ where $S$ is a set of configurations corresponding to
an approximate state estimate; $S_\top \subseteq S$ are the
configurations corresponding to the state estimate which is surely not
approximated; and $r$ is a set of valuations of $Y$. Bad states which
have to be avoided are states where $S_\top$ is empty. Then, a
tba-simulation $\wbs_{\sf det}$ can naturally be defined on this game:

\begin{lemma}\label{lm-det}%[Corollary of Theorem~\ref{th-red}]
  Let $\mathcal{A}$ be a timed automaton, $Y$ be a set of clocks and
  $M$ be a maximal constant for guards.  Then $G_{\mathcal{A},(Y,M)}$
  admits a tba-simulation relation $\wbs_{\sf det}$ defined as
  follows: $((S,S_\top),r) \wbs_{\sf det} ((S',S'_\top),r')$ iff
  $r=r'$, $S \supseteq S'$ and $S_\top \subseteq S_\top$.
\end{lemma}
\begin{proof}
  First of all, $\wbs_{\sf set}$ is clearly a partial order.  It is
  thus sufficient to prove that this partial order is a simulation
  relation.  To do so, we simply prove that moves of Player $A$ and
  Player $B$ preserve the partial order.  Let
  $v^1=((S^1,S^1_\top),r^1)$ and $v^2=((S^2,S^2_\top),r^2)$ such that
  $v_1\wbs_{\sf set}v_2$.
\begin{itemize}
\item {\bf Moves of Player $A$} are resets of clocks.  Let
  $Y'\subseteq Y$ a set of clocks that $A$ can reset.  Then, the
  $Y'$-successors of $v^1$ and $v^2$ are respectively
  $v'^1=((S'^1,S'^1_\top),r'^1)$ and $v'^2=((S'^2,S'^2_\top),r'^2)$,
  where for $i=1$ or $2$, $r'^i=r^i_{[Y'\leftarrow 0]}$, $S'^1=\{{\sf
    conf}_{[Y'\leftarrow 0]} |{\sf conf}\in S^1\}$ and
  $S'^1_\top=\{{\sf conf}_{[Y'\leftarrow 0]} |{\sf conf}\in
  S^1_\top\}$. Then from $v_1\wbs_{\sf set}v_2$, we obtain
  $r'^1=r'^2$, $S'^1\supseteq S'^2$ and $S'^1_\top\subseteq S'^2_\top$
  and hence $v'_1\wbs_{\sf set}v'_2$.
\item {\bf Moves of Player $B$} are an action and a guard over the
  clocks of $Y$.  Let $(a,g)$ be a move of $B$. Then, the respective
  $(a,g)$-successors $v'^1$ and $v'^2$ of $v^1$ and $v^2$ are computed
  thanks to function ${\sf Succ}$ computing set of elementary
  successors of configurations: $v'^1=((S'^1,S'^1_\top),r'^1)$ and
  $v'^2=((S'^2,S'^2_\top),r'^2)$, where for $i=1$ or $2$, $r'^i=r^i$,
  $S'^1= \cup_{{\sf conf}\in S^1} {\sf Succ}({\sf conf},(a,g))$, and
  $S'^1_\top= \cup_{{\sf conf}\in S^1_\top} {\sf Succ}_\top({\sf
    conf},(a,g))$, with ${\sf Succ}_\top$ the function ${\sf Succ}$
  such that only configurations marked $\top$ are kept.  From
  $v_1\wbs_{\sf set}v_2$, we thus obtain that $r'^1=r'^2$,
  $S'^1\supseteq S'^2$ and $S'^1_\top\subseteq S'^2_\top$ and hence
  $v'_1\wbs_{\sf set}v'_2$.
\end{itemize}
\end{proof}

\begin{figure}[t]
\begin{center}
\scalebox{.75}{
\begin{tikzpicture}[->,>=stealth',shorten >=1pt,auto,node distance=2cm,
                     semithick]

   \tikzstyle{every state}=[text=black]

   \node[state, fill=white] [rectangle] (A) {$\begin{array}{ll}  
 \ell_0,x-y=0,\top&\multirow{2}{*}{\{0\}} \\\ell_1,0<x-y<1,\top&  
\end{array}$};
\node(vA) [right of=A, node distance=2.1cm] {$\v_1$};
   \node[state, fill=white] [rectangle] (A0) [above of=A, node distance=3.4cm,  
yshift=0cm] {$\begin{array}{ll} 
\ell_0,x-y=0,\top&\{0\}\end{array}$};
\node(vA0) [right of=A0, node distance=1.8cm] {$\v_0$};
   \node[state, fill=white] (A01) [below of=A0, node distance=1.7cm,  
yshift=0cm] {};   
   \node[state, fill=white] (A1) [below of=A, node distance=1.7cm,  
yshift=0cm] {};   
    \node[state, fill=white] [rectangle] (C)  [right of=A1, node  
distance=0cm, yshift=-1.7cm] {$\begin{array}{ll}  
 \ell_0,x-y=0,\top&\multirow{3}{*}{\{0\}} \\ \ell_1,0<x-y<1,\top&\\\ell_2,x-y=0,\bot&  
\end{array}$}; 
\node(vC) [right of=C, node distance=2.1cm] {$\v_2$};
   \node[state, fill=white] (C1) [left of=C, node distance=4.2cm,  
yshift=1.7cm] {};   

\node[fill=white, color=white] (H) [above of=A0, node  
distance=1.2cm] {};
\node[fill=white, color=white] (H1) [below of=A01, node  
distance=1cm, xshift=-1.5cm] {};
\node[fill=white, color=white] (H2) [below of=A1, node  
distance=.8cm, xshift=1.2cm] {};

\draw(3,2) node [ellipse,draw,minimum height=1.8cm,minimum width=3.2cm,fill=gray!30] (P1) {$\mathcal{G}_{\mathcal{A}_1}$};
\draw(3,-1.8) node [ellipse,draw,minimum height=1.8cm,minimum width=3.2cm,fill=gray!30] (P2) {$\mathcal{G}_{(\mathcal{A}_1,\mathcal{A}_2),(Y,M)}$};

   \path (H) edge node [left] {} (A0)
         (A) edge node [right] {$(0,1),a$} (A1)
         (A1) edge node [left] {$\{y\}$} (C)
         (A1) edge node [above] {$\emptyset$} (H2)
 %       (A1) edge node [above,xshift=-.3cm,yshift=.5cm] {$\emptyset$} (B)
         (A0) edge node [right] {$(0,1),a$} (A01)
         (A01) edge node [right] {$\{y\}$} (A)
         (A01) edge node [above] {$\emptyset$} (H1)
%         (A1') edge node [below] {$\{y\}$} (C)
%        (A1') edge node [below] {$\emptyset$} (B)
        (C) edge [=<,bend left=50,dashed] node [above,sloped,yshift=.2cm] {$\wbs$} (A)
        (C) edge [dashed,line width=2pt] node [below] {$b$} (2,-2)
        (A) edge [dashed,line width=2pt] node [below] {$b$} (2,1.8)
        (C) edge node [sloped,yshift=0cm,above] {$0<x<1,a$} (C1)
        (C1) edge [bend right=20] node [below] {$\{y\}$} (C)
;

  \node[state, fill=white] (Am) [left of=A0, node distance=7cm] {$\ell_0$};
  \node[state, fill=white] (Bm) [below of=Am, node distance=2.5cm] {$\ell_1$};
  %\node[state, fill=white] (Cm) [below of=Am, node distance=3.5cm, yshift=-1cm] {$\ell_2$};
  \node[state, fill=white] (Dm) [below of=Am, node distance=5cm] {$\ell_2$};
  \node[state, fill=white, color=white] (Em) [above of=Am, node distance=1.5cm, yshift=0cm] {};
%\node[state, fill=white, color=white] (Fm) [below of=Dm, node distance=1.5cm, yshift=0cm] {}
\draw(-5,3) node [ellipse,draw,minimum height=1.8cm,minimum width=3.2cm,fill=gray!30] (P1) {$\mathcal{A}_1$};
\draw(-5,0) node [ellipse,draw,minimum height=1.8cm,minimum width=3.2cm,fill=gray!30] (P2) {$\mathcal{A}_2$};
;
  
  \path (Am) edge node [right,xshift=-2cm] {$0<x<1,a$} (Bm)
        (Em) edge (Am)
        (Am) edge [loop left] node [left] {$0<x<1,a$} (Am)
    %    (Am) edge node [below,sloped] {$0<x<1,a,\{x\}$} (Cm)
	(Bm) edge node [right,xshift=-2.5cm] {$0<x<1,a,\{x\}$} (Dm)
%	(Cm) edge node [below,sloped] {$x=0,b$} (Dm)      
  %      (Dm) edge (Fm)
        (Dm) edge [dashed,line width=2pt] node [below,xshift=.1cm] {$b$} (-6,-.2)
        (Bm) edge [dashed,line width=2pt] node [below,xshift=.1cm] {$b$} (-6,2.8)
;
\draw [-,fill=red,color=red]%,line width=1.5pt]
(0.5,1.3) -- (2.3,1) -- (2.3,.95) -- (0.5,1.25) ;
\draw [-,fill=red,color=red]%,line width=1.5pt]
(0.7,1.7) -- (1.8,.7) -- (1.8,.65) -- (0.7,1.65) ;
\end{tikzpicture}}
\caption{Excerpt of an example of game.}
\label{fig:ex-game-beg}
\end{center}
  \vspace{-.7cm}
\end{figure}
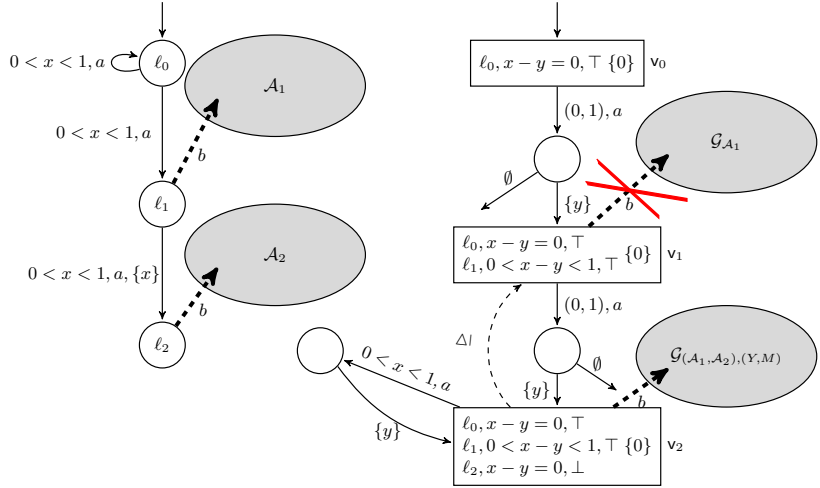

As an example, Figure~\ref{fig:ex-game-beg} presents an excerpt of the
construction of the game (right) for a TA (left), with $Y=\{y\}$.
Gray ellipses symbolise potentially large part of the TA (resp. the
game) which we do not detail here.  To illustrate the tba-simulation,
let us consider the two states $\v_2 = (\{(
\ell_0,x-y=0),(\ell_1,0<x-y<1),(\ell_2,x-y=0)\},\{(
\ell_0,x-y=0),(\ell_1,0<x-y<1)\}, \{0\})$ and $\v_1 = (\{(
\ell_0,x-y=0),(\ell_1,0<x-y<1)\},\{( \ell_0,x-y=0),(\ell_1,0<x-y<1)\},
\{0\})$ in the game. One can easily check that $\v_2\wbs \v_1$. Thanks
to optimisation 3 (see Section~\ref{sec:effic-comp-succ}),
Algorithm~\ref{algonous} applied to this example will avoid exploring
$\mathcal{G}_{\mathcal{A}_1}$

\appendix

\section{The original OTFUR algorithm\label{sec:orig-otfur-algor}}
For the sake of completeness, Algorithm~\ref{algo:otfur_plus} recalls
the original OTFUR algorithm\cite{CDFLL-concur05}, adapted to the case
of safety games.
\begin{algorithm}
 \caption{{\sf OTFUR} \cite{CDFLL-concur05} algorithm for
safety games}\label{algo:otfur_plus} \KwData{G,
$I$}
\tcp{Initialization}
${\sf Passed} := \{I\};$ ${\sf Depend}(I):= \varnothing;$\\
$\textbf{for} \text{ all position }v\textbf{ do } {\sf Losing}[v]:= {\mathsf{false}};$ \\
${\sf Waiting} := \{(I, v')\in E\};$\\

\tcp{Saturation}
\While{${\sf Waiting} \neq \varnothing \wedge \neg {\sf Losing}[I]$}{
$e=(v, v') := {\it pop}({\sf Waiting});$\\
\If{$v'\not\in {\sf Passed}$} {
    ${\sf Passed}:= {\sf Passed}\cup \{v'\};$ \\
    ${\sf Losing}[v']:=  v'\in {\sf Bad};$\\
    ${\sf Depend}[v']:=\{(v,v')\};$\\
  \If{${\sf Losing}[v']$} {
    ${\sf Waiting}:= {\sf Waiting} \cup \{e\};$ \tcp{add $e$ for reevaluation}
   }
  \Else{
    ${\sf Waiting}:= {\sf Waiting} \cup \{(v', v'')\in E\};$\\
  }
}\Else{\tcp{reevaluation}
  ${\sf Losing}^{*} := \begin{array}[t]{ll} & v\in V_A
            \wedge \bigwedge_{v'',(v,v'')\in E}{\sf Losing}[v'']\\
            \vee & v\in V_B\wedge \bigvee_{v'',(v,v'')\in E}{\sf Losing}[v''];
        \end{array}$\\
        \If{${\sf Losing}^*$}{
            ${\sf Losing}[v]:= \mathsf{true}$;\\
          ${\sf Waiting}:= {\sf Waiting} \cup {\sf
            Depend}[v]$\tcp{back propagation}
        }
        \textbf{if} $\neg {\sf Losing}[v']$ \textbf{then} ${\sf Depend}[v']:= {\sf Depend}[v']\cup\{e\}$
} }
\KwRet{$\neg {\sf Losing}[I]$}
\end{algorithm}

%%% Local Variables: 
%%% mode: latex
%%% TeX-master: "main"
%%% End: 

\end{document}